\documentclass[nofootinbib,twocolumn,pra,letterpaper,longbibliography]{revtex4-2}
\usepackage{latexsym,amsmath,amssymb,amsfonts,graphicx,color,amsthm}
\usepackage{enumerate}
\usepackage{braket}
\usepackage{adjustbox}
\usepackage{footnote}
\usepackage[dvipsnames]{xcolor}
\usepackage{tipa}
\usepackage{textcomp}
\usepackage{fontenc}
\usepackage{mathrsfs}
\usepackage{color}
\usepackage{colortbl}
\usepackage{graphics,graphicx}
\usepackage{txfonts}
\usepackage{lipsum, babel}

\usepackage[unicode=true,pdfusetitle, bookmarks=true,bookmarksnumbered=false,bookmarksopen=false, breaklinks=false,pdfborder={0 0 0},backref=false,colorlinks=false] {hyperref}
\hypersetup{ colorlinks,linkcolor=myurlcolor,citecolor=myurlcolor,urlcolor=myurlcolor}

\definecolor{myurlcolor}{rgb}{0,0,0.7}

\usepackage{float}
\usepackage{hyperref}
\usepackage{subcaption}

\newcommand{\cA}{\mathcal{A}}

\newcommand{\cC}{\mathcal{C}}

\newcommand{\cE}{\mathcal{E}}

\newcommand{\cG}{\mathcal{G}}
\newcommand{\cH}{\mathcal{H}}
\newcommand{\cI}{\mathcal{I}}
\newcommand{\cJ}{\mathcal{J}}
\newcommand{\cK}{\mathcal{K}}
\newcommand{\cL}{\mathcal{L}}
\newcommand{\cM}{\mathcal{M}}

\newcommand{\cS}{\mathcal{S}}
\newcommand{\cT}{\mathcal{T}}
\newcommand{\cU}{\mathcal{U}}

\newcommand{\Id}{\mathbb{I}}
\newcommand{\tr}{\text{Tr}}

\newtheorem{theorem}{Theorem}

\newtheorem{corollary}[theorem]{Corollary}

\newtheorem{example}{Example}

\newtheorem{obs}{Observation}

\begin{document}

\title{Improvement in quantum communication using quantum switch}
\author{Arindam Mitra$^{1,2}$}
\email{amitra@imsc.res.in}
\email{arindammitra143@gmail.com}
\affiliation{$^1$Optics and Quantum Information Group, The Institute of Mathematical Sciences,
C. I. T. Campus, Taramani, Chennai 600113, India.\\
$^2$Homi Bhabha National Institute, Training School Complex, Anushaktinagar, Mumbai 400094, India.}
\author{Himanshu Badhani$^{1,2}$}
\email{himanshub@imsc.res.in}
\email{himanshubadhani@gmail.com}

\affiliation{$^1$Optics and Quantum Information Group, The Institute of Mathematical Sciences,
C. I. T. Campus, Taramani, Chennai 600113, India.\\
$^2$Homi Bhabha National Institute, Training School Complex, Anushaktinagar, Mumbai 400094, India.}
\author{Sibasish Ghosh$^{1,2}$}
\email{sibasish@imsc.res.in}
\affiliation{$^1$Optics and Quantum Information Group, The Institute of Mathematical Sciences,
C. I. T. Campus, Taramani, Chennai 600113, India.\\
$^2$Homi Bhabha National Institute, Training School Complex, Anushaktinagar, Mumbai 400094, India.}

\date{}

\begin{abstract}
Applications of the quantum switch on quantum channels have recently become a topic of intense discussion. In the present work, we show that some useless (for communication) channels may provide useful communication under the action of quantum switch for several information-theoretic tasks: quantum random access codes, quantum steering, etc. We demonstrate that the quantum switch can also be useful in preventing the loss of coherence in a system when only coherence-breaking channels are the available channels for communication. We also show that if a useless quantum channel does not provide useful communication even after using a quantum switch, concatenating the channel with another suitable quantum channel, and subsequently using the switch, one may achieve useful communication. Finally, we discuss how the introduction of noise in the quantum switch can reduce the advantage that the switch provides.

\end{abstract}

\maketitle
\section{Introduction}
Perfect communication in the quantum world is impossible as inevitable noise gets introduced in the communication channel due to the interaction of any quantum system with its environments. It is, therefore, important to explore the possibility of better communication using such imperfect quantum channels.

The idea of a `quantum switch'  \cite{Chi13} is a process -- used in the context of indefinite causal order -- through which superposition of concatenations of two quantum channels, taken in two different orders, is realized. The superposition is created with the help of the coupling to a two-level ancilla (called the control qubit). In this way, the switch creates an \textit{indefinite causal order}, between two channels \cite{costa12}.
Quantum switch has several applications in different information-theoretic and thermodynamics tasks. More specifically, quantum switch has been used to test properties of quantum channels \cite{Chiribella-perf-disc,Costa-comput}, to reduce quantum
communication complexity \cite{Brukner-complex}, to boost the precision of quantum metrology \cite{Zhao20}, and to achieve thermodynamics advantages \cite{Guha20} etc. Recently, it was shown that using the quantum switch on a specific noisy channel, along with controlled operations, one can achieve perfect communication \cite{chiribella_perfect}. However, that example of the channel is unique up to unitary transformations.
\\
{In this context, one may raise the following pertinent questions: (1) If a useless channel (i.e. a channel which does not perform useful communication for a \emph{given} information-theoretic task) does not provide improvement in communication even after the use of quantum switch then how to make those channel useful -(for that particular information theoretic task) under the action of quantum switch? (2) How much improvement can be achieved by applying quantum switch on useless quantum channels for \emph{different} information-theoretic tasks?

 The present work is in the direction of providing answers to some of these questions (specially in Sec. \ref{subsec:conc_ch} and Sec. \ref{subsec:Advantages_information} respectively).} The advantages of using the quantum switch in quantum communications have been studied before either in the context of certain scenarios like teleportation \citep{Pati20} or communications via several noisy channels \citep{Mcl20,Proc19,Saz21}. In Sec. \ref{subsec:Advantages_information} of this work, we discuss whether a quantum switch can improve the capability of a quantum channel to transfer quantum resources like entanglement, steerability, coherence, etc.
It depends on the information-theoretic tasks whether a channel is useless for that particular task. Now consider the case where such a useless channel is unavoidable for communication due to some constraints e.g., spatial distances, type of interaction with the environment, etc. Now, it is possible that such a useless channel remains useless \emph{even after} the use of quantum switch on it. We call such channels as a \emph{completely useless} channels. Now, the question is whether there exists any process which makes these channels useful under the action of quantum switch. In Sec. \ref{subsec:conc_ch}, we show that a possible answer to that question is the concatenation of a quantum channel with that useless channel. More specifically, we show that if at first such a completely useless channel is concatenated with another quantum channel and after that, the switch is used on the resulting channel then useful communication may be achieved.
\\
It has been found that  using the superposition of trajectories can provide advantages over a single trajectory. Specifically, it has been shown that the so-called superposition of direct pure processes (SDPPs) \cite{Abbot2020,PRA-Rubino} can also provide perfect communication under a noisy channel. These components of the superposition of the trajectories need not have two different causal orders like what we get from a quantum switch, thus making the quantum switch a subset of SDPPs. It has therefore been argued that the advantages of the quantum switch should be attributed to the coherence in the path degree of freedom and not to the indefinite causal order. In these cases, however, the advantages arise due to the use of the external degree  (or the path degree) of freedom to carry information \citep{shannon, resource}. In the case of the quantum switch, however, the information is not carried in the path degree of freedom. Recently, the effect of these more generalized superposition of trajectories on the four specific noisy qubit quantum channels has been also studied in \cite{PRR-Rubino} in the context of quantum capacity. However, the quantum switch is placed in a distinct category of its own as uses the indefinite causal order to gain advantages in communication. For this reason, we will restrict ourselves to the study of the improvements in communications under the action of the quantum switch compared to when the quantum switch is not available or used. We will not bring other processes like SDPP in this paper as done in some other works \citep{Pati20, Proc19, Saz21, Mcl20}.
\\
The rest of the paper is organized as follows. In Sec. \ref{sec:prelims}, we discuss some preliminaries. In Sec. \ref{sec:main}, we discuss our main results. {In particular, in Subsec. \ref{subsec:transf-en}, we argue that it is important to study the effect of quantum switch considering one output branch at a time. In Subsec. \ref{subsec:conc_ch}, we discuss that a completely useless channel can be converted into a useful channel through the concatenation of it with another quantum channel and after that the use of switch on the resulting channel.} In Subsec. \ref{subsec:Advantages_information}, we discuss how communication improvement helps to get the advantage in some of the well-known information-theoretic tasks. In Subsec.\ref{subsec:noisy switch}, we discuss the fact that the aforesaid improvement in communication can decrease if the switch is noisy. Finally, in Sec. \ref{sec:conc}, we summarise our results and discuss some future directions. 
\section{Preliminaries}\label{sec:prelims}
In this section, we discuss the preliminaries, to be used in the latter sections of the present paper.
\subsection{Observables and their compatibility}
An observable $A$, acting on Hilbert space $\cH$ of dimension $d$, is defined as a set of positive semi-definite matrices $\{A(x)\}$ on $\cH$ such that $\sum_x A(x)=\Id$. We denote the outcome set of $A$ as $\Omega_A$ and therefore, in the above, $x \in {\Omega}_A$. If all $A(x)$'s are projectors then $A$ is a PVM and otherwise $A$ is a POVM. Now two observables $A$ and $B$ are said to be compatible if there exists a joint observable $\cG=\cG(x,y)$ with outcome set $\Omega_A\times\Omega_B$ such that 
\begin{align}
A(y)=\sum_y\cG(x,y);B(y)=\sum_x\cG(x,y)
\end{align}
 hold for all $x\in \Omega_A$ and $y\in\Omega_B$ \citep{Hzi12,Hein08,Hein16}. The above definition is equivalent to the statement that two observables $A_1$ and $A_2$ are compatible if there exists a probability distribution $P(x_i|i,\lambda)$ and an observable $\cJ=\{J_{\lambda}\}_{\lambda\in\Omega_{\cJ}}$ with outcome set ${\Omega}_J$, such that $A_i(x_i)=\sum_{\lambda}P(x_i|i,\lambda)J_{\lambda}$ for all $i\in\{1,2\}$. Implementation of the joint observable allows simultaneous implementation of both observables. If such joint observable does not exist for a set of observables, then the set is incompatible. Commutativity of observables always imply compatibility, but compatibility of observables implies commutativity only for PVMs.

\subsection{Quantum channels}
Let us denote the state space on Hilbert space $\cH$ as  $\cS(\cH)$. Quantum channels are the CPTP maps from $\Lambda:\cS(\cH_1)\rightarrow\cS(\cH_2)$ that satisfies the equation $\Lambda(\sum_ip_i\rho_i)=\sum_ip_i\Lambda(\rho_i)$ for any arbitrary probability distribution $\{p_i\}$ and for any arbitrary set of states $\{\rho_i\in\cS(\cH_1)\}$ \citep{Nielson}. Let $\cL(\cH_2)$ be the set of all bounded linear operators acting on the Hilbert space $\cH$. A CP unital map $\Lambda^*:\cL(\cH_2)\rightarrow\cL(\cH_1)$ is the dual channel of $\Lambda$ if $\tr[\Lambda(\rho) A(x)]=\tr[\rho\Lambda^*(A(x))]$ for all $x\in\Omega_A$ and  for all $\rho\in \cS(\cH_1)$ with $A = \{A(x)\}_{ x \in {\Omega}_A}$ being an arbitrary observable acting on ${\cal H}_2$. It is well known that any quantum channel $\Lambda$ admits Krauss representation such that $\Lambda(\rho)=\sum_x \cK_x\rho\cK^{\dagger}_x$ where $\sum_x\cK^{\dagger}_x\cK_x=\Id$. $\cK$'s are the Krauss operators of $\Lambda$. A channel is called unital if it keeps the maximally mixed state unchanged.
We know that any qubit state can be represented by a three component real vector $\vec{a} = (a_1, a_2, a_3)$ as 
$$
\rho=\dfrac{1}{2}(I+\vec{a}.\vec{\sigma})
$$ 
where $\vec{\sigma}=(\sigma_x,\sigma_y,\sigma_z)^T$ are the Pauli matrices and $|\vec{a}| \leq 1$. Any map from a qubit state to a qubit state can therefore be represented by a linear map $T$ such that 
\begin{equation}
\begin{aligned}
	\Gamma(\rho)=&\dfrac{1}{2}(I+\vec{a}'.\vec{\sigma})\\
	=&\dfrac{1}{2}(I+(T\vec{a}+\vec{t}).\vec{\sigma})
\end{aligned}
\end{equation}
This map $\Gamma$ is generally represented in form of a matrix (known as $T$-matrix):
\begin{equation}
\begin{aligned}
\mathcal{T}_{\Gamma}=\begin{pmatrix}
1 & 0 \\
\vec{t} & T \\
\end{pmatrix}
\end{aligned}
\end{equation}
where for  given any qubit state $\rho$ represented as a column vector $v_{\rho}=(1,a_1,a_2,a_3)^T$, the action of the map is given by a matrix multiplication $v_{\Gamma(\rho)}=\mathcal{T}_{\Gamma}.v_{\rho}$ \cite{Hzi12}. The qubit map is unital when the vector $\vec{t}=0$. We should mention here that we index the components of $v_{\rho}$ from $0$ to $3$.
\subsection{Special types of quantum channels}\label{subsec:special-type}
\subsubsection{Depolarizing channels}
A quantum depolarizing channel $\Gamma^t_d:\cS(\cH)\rightarrow\cS(\cH)$ is a specific type of quantum channel which has following form:
\begin{align}
\Gamma^t_d(\rho)=t\rho+(1-t)\frac{\Id}{d} \label{depchannel}
\end{align}
where $-\frac{1}{d^2-1}\leq t\leq 1$ and $d$ is the dimension of the Hilbert space ${\cal H}$. For qubit depolarizing channels, a set of krauss operators is$\{\sqrt{1-p}\Id,\sqrt{\frac{p}{3}}\sigma_x, \sqrt{\frac{p}{3}}\sigma_y, \sqrt{\frac{p}{3}}\sigma_z\}$ . Here, $t=1-\frac{4p}{3}$.  Clearly, quantum depolarizing channels are unital.
\subsubsection{Entanglement breaking channels}
A quantum channel $\Lambda^E:\cS(\cH_B)\rightarrow\cS(\cK)$ is called an entanglement breaking channel (EBC) if for all $\rho_{AB}\in\cS(\cH_A\otimes\cH_B)$, $(\cI\otimes\Lambda^E)(\rho_{AB})$ is a separable state (irrespective of the dimension of ${\cal H}_A)$ \citep{Hor03,Rus03}. ${\Lambda}^E$ is an EBC iff its Choi matrix is separable. Entanglement breaking channels form a convex set.  Any EBC $\Lambda^E$ has the following form:
\begin{equation}
\Lambda^E(\rho)=\sum_x\rho_x\tr[\rho A(x)]
\end{equation}
where $A=\{A(x)\}$ is a POVM acting on $\cH_B$ and ${\rho}_x \in {\cal S}({\cal K})$. $\Gamma^t_2$ (given in equation \eqref{depchannel}) is EBC iff $\mid t\mid\leq \frac{1}{3}$. We denote set of all EBCs acting on the state space $\cS(\cH)$ (where $d$ dimensional Hilbert space $\cH$) as $\cC^d_{EBC}$. It is well known that $\Phi\circ\Lambda^{EBC}$ is an $EBC$ for any quantum channel $\Phi$. If Choi matrix of a CP linear map is separable then we will call it an EB CP map.
\subsubsection{Incompatibility breaking channels}
A quantum channel $\Lambda^{n-IBC}$ is an $n$-incompatibility breaking channel($n$-IBC) if for an arbitrary set of $n$ observables $\{A_1,.....,A_n\}$, the set $\{(\Lambda^{n-IBC})^*(A_1),.....,(\Lambda^{n-IBC})^*(A_n)\}$ is compatible \cite{Heino_IBC}. The set of all $n-IBC$'s form a convex set. $\Gamma^t_d$ is $n-IBC$ for $t\leq \frac{n+d}{n(1+d)}$ \cite{Heino_IBC}. We denote the set of all $n-IBC$s acting on the state space $\cS(\cH)$ (where $d$ dimensional Hilbert space $\cH$) as $\cC^d_{n-IBC}$. It is well known that $\Phi\circ\Lambda^{n-IBC}$ is a $n-IBC$ for any quantum channel $\Phi$.\\
A quantum channel $\Lambda^{IBC}$ is IBC if it is $n-IBC$ for all $n$ \cite{Heino_IBC}. $\Gamma^t_d$ is $IBC$ for $t\leq \frac{(3d-1)(d-1)^{(d-1)}}{d^d(d+1)}$ \cite{Heino_IBC}. We denote set of all IBCs acting on the state space $\cS(\cH)$ as $\cC^d_{IBC}$. It is well known that $\cC^d_{IBC}\subseteq.....\subseteq \cC^d_{(n+1)-IBC}\subseteq \cC^d_{n-IBC}\subseteq .....\subseteq \cC^d_{2-IBC}$.  It is also well known that $\cC^d_{EBC}\subset\cC^d_{IBC}$. It is also well known that $\Phi\circ\Lambda^{IBC}$ is an $IBC$ for any quantum channel $\Phi$.
\subsection{Quantum switch}
A causal order is a partial ordering of events in the space-time. In our case the \textit{events} will correspond to quantum operations which are applied on a given system at two different times and places. Let a  quantum system be sent from Alice to Bob through a quantum channel. The quantum channel is composed of two channels $\Lambda$ and $\Phi$. Conventional understanding dictates that resulting channel can be either of the form $\Lambda\circ \Phi$ or $\Phi\circ \Lambda$ or a classical mixture of the two orders, i.e. $p\Lambda\circ \Phi+(1-p)\Phi\circ \Lambda$ with $0\leq p\leq 1$. Causal ordering between two events is therefore the specification of which event effects the other. However, Oreshkov et al \cite{costa12} have shown that by using the property of quantum superposition one can construct processes wherein the order of events cannot be specified. As a result, if Alice has at her disposal two quantum channels $\Lambda$ and $\Phi$, she can construct a third channel which cannot be seen as a simple composition of the two given channels. This channel can be seen as a superposition of the two orderings of the initial two channels $\Lambda$ and $\Phi$ and therefore has the \textit{indefinite} causal order between the two channels.

Quantum switch was introduced \citep{Chi13,Moqanaki-15,robino17,Gia18} as a quantum circuit that implements the indefinite causal order between two different operations. The basic ingredient for such a processes is a two level ancilla systems  (We call it control qubit) with the state $\omega$. The signalling between the two events depends on the state of the ancilla and using the fact that the ancilla can be in a superposition of more than one state, the order of the two events can be in a superposition of the two scenarios: $\Lambda\circ\Phi$ or $\Phi\circ\Lambda$.\color{black}  We start with the product state $\rho\otimes\omega$ where $\rho$ is the state of the system on which these operators ({\it i.e.} channels) act and $\omega$ is a state of control qubit. Given two quantum channels $\Lambda$ and $\Phi$ the quantum switch produces an output $(\Lambda\circ \Phi)(\rho)$ if $\omega=\ket{0}\bra{0}$ and $(\Phi\circ \Lambda)(\rho)$ if $\omega=\ket{1}\bra{1}$. If $\omega$ is in a superposition of states $\ket{0}$ and $\ket{1}$, the resultant operation on $\rho$ will be in a superposition of $(\Lambda\circ \Phi)(\rho)$ and $(\Phi\circ \Lambda)(\rho)$ thus mimicking the indefinite causal order between the two operations,  (see fig. \ref{fig:switch})
\\
\begin{figure}[hbt!]
\includegraphics[scale=.12]{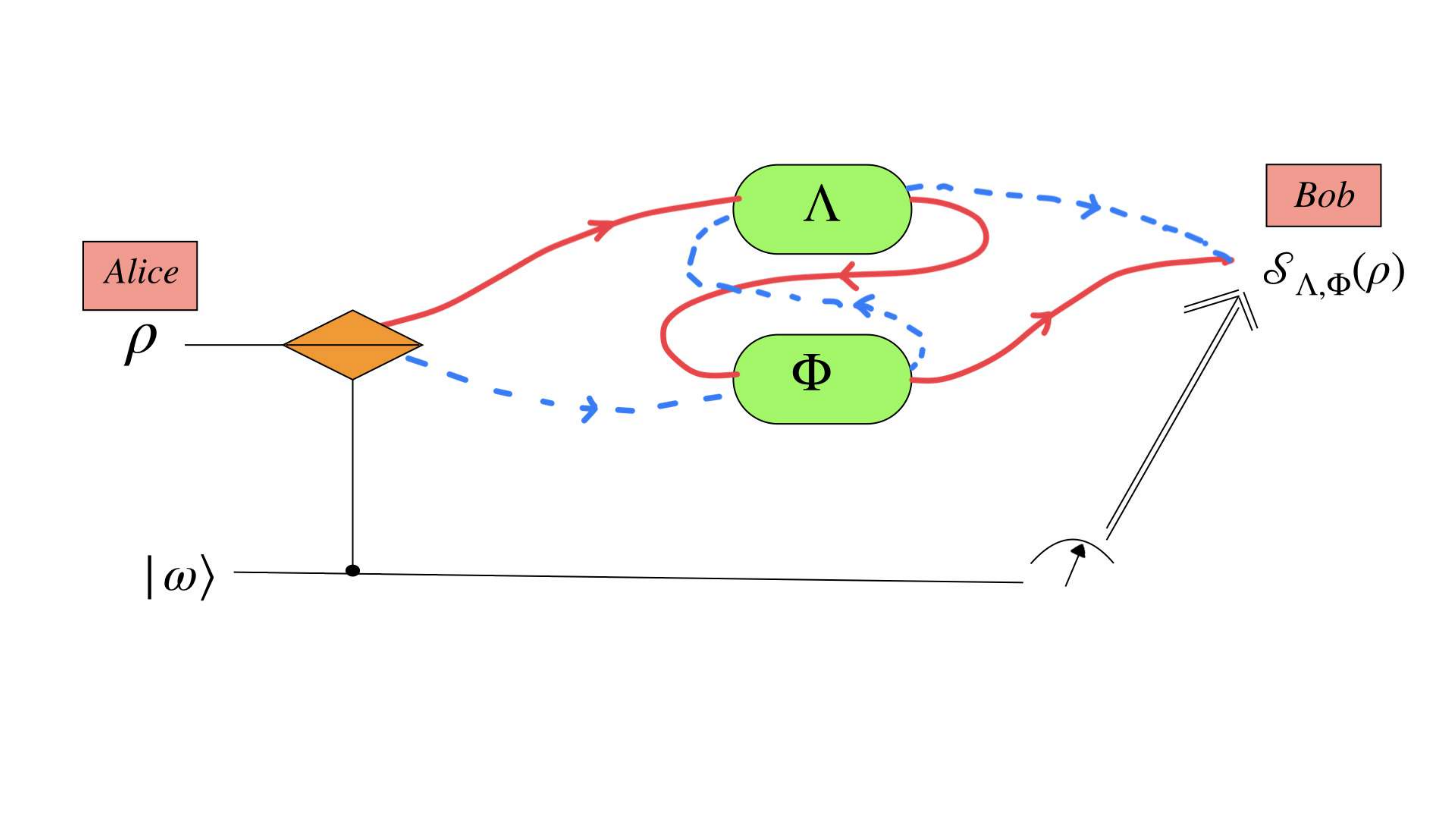}
\caption{{A quantum switch determines the order of two quantum operations acting on a state $\rho$ conditional on the state of the system $\omega$. The two causal orders are denoted by the red (bold) and blue (dotted) lines. If the state $\omega$ is in a superposition of the two basis states the resulting channel will also be in a superposition of two causal orders}.}
\label{fig:switch}
\end{figure}
\\ \\
Let the set of Krauss operators of channels $\Lambda$ and $\Phi$ be $\{\mathcal{A}_x\}$ and $\{\mathcal{B}_x\}$ respectively. The Krauss operators for the Switch are given by:
\begin{equation}
\begin{aligned}
\mathsf{S}_{x,y}=\mathcal{A}_x\mathcal{B}_y\otimes\ket{0}\bra{0}+\mathcal{B}_y\mathcal{A}_x\otimes\ket{1}\bra{1}
\end{aligned}
\end{equation}
Using $\ket{0}\bra{0}=(\sigma_z+I)/2$ and $\ket{1}\bra{1}=(I-\sigma_z)/2$, the above Krauss operators take the following form
\begin{equation}
\begin{aligned}
	\mathsf{S}_{x,y}=\dfrac{1}{2}(\mathcal{A}_x\mathcal{B}_y+\mathcal{B}_y\mathcal{A}_x)\otimes I+\dfrac{1}{2}(\mathcal{A}_x\mathcal{B}_y-\mathcal{B}_y\mathcal{A}_x)\otimes\sigma_z
	\end{aligned}
\end{equation}
If we input identical channels into the switch, that is $\Lambda=\Phi$, the resultant channel has the following form:
\begin{align}
	&\mathsf{S}_{\Lambda,\Lambda,\omega}(\rho)\nonumber\\
	&=\dfrac{1}{4}\sum_{x,y}\Big(\{\mathcal{A}_x,\mathcal{A}_y\}\rho\{\mathcal{A}_x,\mathcal{A}_y\}^{\dagger}\otimes\omega\nonumber\\
	& +[\mathcal{A}_x,\mathcal{A}_y]\rho[\mathcal{A}_x,\mathcal{A}_y]^{\dagger}\otimes\sigma_z\omega\sigma_z\Big)\nonumber\\
	&\equiv C_{+,\Lambda,\Lambda}(\rho)\otimes\omega+C_{-,\Lambda,\Lambda}(\rho)\otimes\sigma_z\omega\sigma_z
	\label{switched_channel}
	\end{align}
 
In general, $C_{+,\Lambda,\Lambda}$ and $C_{-,\Lambda,\Lambda}$ are CP maps, but are not trace preserving and hence are not quantum channels. We will occasionally refer to these branches as the $C_+$ and $C_-$ branches respectively for a general quantum channel $\Lambda$. As given in \cite{chiribella_perfect}, for a qubit channel of the form $\Lambda=\Phi=\sum_{\mu=0}^3p_\mu\sigma_\mu\rho\sigma_\mu$ (i.e., for Pauli channels), the two branches $C_{+,\Lambda_A,\Lambda_A}(\rho)$ and $C_{-,\Lambda,\Lambda}(\rho)$  become $C_{+,\Lambda,\Lambda}(\rho)=q\bar{C}_{\Lambda,+}(\rho)$ and $C_{-,\Lambda,\Lambda}(\rho)=(1-q)\bar{C}_{\Lambda,-}(\rho)$ where $0\le q\le 1$  and $\bar{C}_{\Lambda,+}(\rho)$ and $\bar{C}_{\Lambda,-}(\rho)$ are quantum channels. The exact forms of $\bar{C}_{\Lambda,+}(\rho)$ and $\bar{C}_{\Lambda,-}(\rho)$ are given by

\begin{align}
	\bar{C}_{\Lambda,+}(\rho)=&[(p_0^2+p_1^2+p_2^2+p_3^2)\rho\nonumber\\
	&+2p_0(p_1 \sigma_x\rho\sigma_x + p_2 \sigma_y\rho\sigma_y + p_3\sigma_z\rho\sigma_z)]/q,\label{1st_branch_random}\\[.2cm]
	\bar{C}_{\Lambda,-}(\rho)=&[2p_1 p_2 \sigma_z\rho\sigma_z + 2p_2 p_3 \sigma_x\rho\sigma_x\nonumber\\
	& \hspace{2.5cm}+ 2p_1 p_3 \sigma_y\rho\sigma_y]/(1-q)\label{2nd_branch_random}\\[.2cm]
	&\text{where}\hspace{.5cm} q=1-2(p_1p_2+p_2p_3+p_3p_1).
\end{align}
If $\Lambda$ is a Pauli channel, sometimes, we write $\bar{C}_{\Lambda,\pm}$ as $\bar{C}_{\pm}$ where there is no confusion. One can therefore interpret the action of the switch on a Pauli channel as a convex combination of two channels $\bar{C}_{\Lambda,+}$ and $\bar{C}_{\Lambda,-}$ with $q$ interpreted as the probability of getting the $\bar{C}_{\Lambda,+}\otimes \omega$ state. 

When we say ``the action of quantum switch on the quantum channel $\Lambda$", we mean the circuit described in Fig. \ref{fig:switch} with $\Phi=\Lambda$.
\\
Quantum switch can be generalized with the help of an $N$-level control system. Using this $N$-level control system, one can superpose the different causal orders of more than two channels \cite{Proc19} and get an advantage in transfer of information. However, this advantage does not seem to scale well with the increase in quantum resources required to create the coherence in the higher level control system. From this perspective, such a generalization of the quantum switch does not seem to lead to a drastically more efficient result. For this reason, we restrict ourselves to a $2$-level control system in this work.

\section{Main Results}\label{sec:main}
Suppose, Alice has system prepared in a bipartite qubit maximally entangled state. She wants to send one part to Bob so that these shared bipartite entangled states can be used later in information-theoretic tasks. But Alice does not have control over the environment and the only channels available to her are  EBCs. Now if Alice sends these states through an EBC, the resulting bipartite state will no longer be an entangled state and therefore cannot be used for later information-theoretic tasks which require entanglement. Similarly, if Alice has a set of IBCs, she can transfer some entanglement, but those entangled states will not be useful to demonstrate information-theoretic tasks like- quantum steering (it will be discussed later in Sec. \ref{subsubsec:adv-steer} ) or Bell non-locality. By ``transfer of entanglement" we mean the transfer of one part of an entangled bipartite state.  This phenomenon is also known as sharing/distribution of entanglement. If the channel used to transfer this part is an EBC, the transfer of entanglement is not possible. As mentioned before, Ref. \cite{chiribella_perfect} shows that the quantum switch can facilitate a perfect transfer of state if an entanglement breaking Pauli channel of the form $\Lambda(\rho)=(\sigma_x\rho\sigma_x+\sigma_y\rho\sigma_y)/2$ is used. This is an example of the use of the quantum switch in preserving a quantum resource even when the only operation available is a resource-destroying map. But unfortunately, such a channel is unique up to a unitary equivalence \cite{chiribella_perfect}. Therefore the question arises: to what extent the quantum switch can be used to provide an advantage in the transfer of quantum resources when we only have access to the resource-destroying maps? Here, we focus mainly on the transfer of entanglement, steerability, and coherence from Alice to Bob using a quantum switch and an EBC, an IBC, and a coherence-breaking channel respectively. We also focus on how communication using quantum switch provide an advantage in $(n,d)-$quantum random access codes (it will be discussed in Sec. \ref{subsubsec:adv-QRAC}) when $n-$IBCs are the only available channels for Alice to communicate with Bob. We restrict ourselves to the switch operation on two identical quantum channels throughout the paper. We will show that the advantage of the switch over a single use of the channel $\Lambda$ is sometimes constrained to only one of the branches (either $C_{\Lambda,+}$ or $C_{\Lambda,+}$), in which case one can post-select that particular branch. In some cases, however, the advantage of switch manifests in both branches resulting in a deterministic advantage.
\\
 We would like to mention here that throughout this paper we repeatedly use the term ``useless channels" which is context-dependent e.g., EBCs are useless in the context of entanglement transfer and IBCs are useless in the context of QRAC or quantum steering (it will be discussed later in the relevant sections). The channels which are not useless, depending on the context, are considered here to be useful.
\subsection{Transfer of entanglement through EBC using quantum switch}\label{subsec:transf-en}
In this section, we present the advantage of using quantum switch on EBCs. As we have seen in equation \eqref{switched_channel}, when the same channel is used in the switch operation, the resulting operation may be seen as a superposition of two CP maps $C_{+,\Lambda_A,\Lambda_A}$ and $C_{-,\Lambda_A,\Lambda_A}$. We make the following theorem regarding these two branches of the switch:
\begin{theorem}
Assume that Alice has a quantum channel ${\Lambda}_A$ to send a quantum state. If the communication is supported by a quantum switch, i.e. instead of $\Lambda_A$ Alice uses $\mathsf{S}_{\Lambda_A,\Lambda_A,\omega}$ as given in equation \eqref{switched_channel}, then the following will hold good (in the context of quantum switch)-\\

(a) If both the branches $C_{+, {\Lambda}_A, {\Lambda}_A}$ and $C_{-, {\Lambda}_A, {\Lambda}_A}$ are EB CP maps, there does not exist any quantum measurement based control operation which can make the final channel (after tracing out the control part) a non-EBC.

(b) In case, $\Lambda_A$ is a Pauli channel, if both the branches $\bar{C}_{ {\Lambda}_A,+}$ and $\bar{C}_{{\Lambda}_A,+}$ are IBCs, there does not exist any quantum measurement based controlled operation which can make the final channel (after tracing out the control part) a non-IBC.\\
\end{theorem} 

\begin{proof}
(a) Suppose after the switch operation a quantum measurement based controlled operation has been implemented which has the form-
$\mathscr{U}(\rho\otimes \omega)=\Lambda_1(\rho)\otimes\bra{a}\omega\ket{a}\ket{a}\bra{a}+\Lambda_2(\rho)\otimes\bra{a^{\perp}}\omega\ket{a^{\perp}}\ket{a^{\perp}}\bra{a^{\perp}}$
where $(|a\rangle, |a^{\bot}\rangle)$ is an arbitrary pair of mutually orthogonal states in a two dimensional Hilbert space and $\Lambda_1$ and $\Lambda_2$ are the arbitrary channels. This operation $\mathscr{U}$ can be implemented by measuring the observable $\cA=\{\ket{a}\bra{a}, \ket{a^{\perp}}\bra{a^{\perp}}\}$ on the state $\omega$ and then implementing the channels $\Lambda_1$ and $\Lambda_2$ on the state $\rho$ after getting outcomes $a$ and $a^{\perp}$ respectively. Then from equation \eqref{switched_channel}, the resulting state (after the switch operation and then implementing $\mathscr{U}$) will be given by
\begin{align}
\mathscr{U} (\mathsf{S}_{\Lambda_A,\Lambda_A}(\rho))=&(\Lambda_1\circ C_{+,\Lambda_A,\Lambda_A})(\rho)\otimes\bra{a}\omega\ket{a}\ket{a}\bra{a}+\nonumber\\
&(\Lambda_1\circ C_{-,\Lambda_A,\Lambda_A})(\rho)\otimes\bra{a}\sigma_z\omega\sigma_z\ket{a}\ket{a}\bra{a}+\nonumber\\
&(\Lambda_2\circ C_{+,\Lambda_A,\Lambda_A})(\rho)\otimes\bra{a^{\perp}}\omega\ket{a^{\perp}}\ket{a^{\perp}}\bra{a^{\perp}}+\nonumber\\
&(\Lambda_2\circ C_{-,\Lambda_A,\Lambda_A})(\rho)\otimes\bra{a^{\perp}}\sigma_z\omega\sigma_z\ket{a^{\perp}}\ket{a^{\perp}}\bra{a^{\perp}}
\end{align}
Therefore, the effective channel on the system after tracing out the control part is 
\begin{align}
\tr_{\cH_c}[\mathscr{U} (\mathsf{S}_{\Lambda_A,\Lambda_A}(\rho))]=&(\Lambda_1\circ C_{+,\Lambda_A,\Lambda_A})(\rho)\bra{a}\omega\ket{a}+\nonumber\\
&(\Lambda_1\circ C_{-,\Lambda_A,\Lambda_A})(\rho)\bra{a}\sigma_z\omega\sigma_z\ket{a}+\nonumber\\
&(\Lambda_2\circ C_{+,\Lambda_A,\Lambda_A})(\rho)\bra{a^{\perp}}\omega\ket{a^{\perp}}+\nonumber\\
&(\Lambda_2\circ C_{-,\Lambda_A,\Lambda_A})(\rho)\bra{a^{\perp}}\sigma_z\omega\sigma_z\ket{a^{\perp}}
\end{align}
where $\cH_c$ is the Hilbert space of control qubit.

Since, $C_{+,\Lambda_A,\Lambda_A}$ and $C_{-,\Lambda_A,\Lambda_A}$ are EB CP maps, concatenation of any quantum channel with these CP maps provide EB CP maps. Since the linear combination maps (with non-negative coefficients) of EB CP maps is an EB CP map, we have proved that $\tr_{\cH_c}[\cU (\mathsf{S}_{\Lambda_A,\Lambda_A,\omega})]$ is an EBC (since it is also trace preserving).\\
Generalisation of this proof to POVM based control operations is straightforward.

(b) The proof of (b) is similar to that of (a).
\end{proof}
Therefore, the use of quantum switch will not give any significant benefit in communication for above-said channels.

Now, suppose one of the branches, say for example $C_{+,\Lambda_A,\Lambda_A}$ is not an EB CP, even then it is possible that there does not exist any quantum controlled operation that can make the effective channel non-EBC (since the other branch is EB CP). But in that case Alice can improve communication probabilistically. In fact, in that case, Alice can perform measurements on the control qubit (if the state of the control qubit $\omega=\ket{+}\bra{+}$ where $\ket{+}=\frac{\ket{0}+\ket{1}}{\sqrt{2}}$ and $\{\ket{0},\ket{1}\}$ is the eigen basis of $\sigma_z$ then Alice can perform measurements in the basis $\{\omega, \omega^{\perp}\}$ where $\omega^{\perp}=\ket{-}\bra{-}$ with $\ket{-}=\frac{\ket{0}-\ket{1}}{2}$) and will record the outcomes. If the outcome corresponds to $\omega$ (corresponding to non-EB CP branch), the resulting post-measurement state will be non-separable (depends on the input entangled state), in general, otherwise she will throw away the output state.
If both branches (i.e., both $C_{+, {\Lambda}_A, {\Lambda}_A}$ and $C_{-, {\Lambda}_A, {\Lambda}_A}$ ) are not EB CP maps then clearly, it is possible to improve communication deterministically. Therefore, the above-said discussion suggests that it is important to study the effect of quantum switch on a quantum channel considering \emph{one individual output branch} (i.e., either $C_{+, {\Lambda}_A, {\Lambda}_A}$ or $C_{-, {\Lambda}_A, {\Lambda}_A}$ ) at a time.

Let us study the aforesaid discussion through the example of any unital qubit entanglement breaking channel. As discussed before, any unital qubit channel corresponds to the $T-matrix=diag(\lambda_1,\lambda_2,\lambda_3)$ upto the action of unitary operators on the input and output spaces. But the actions of unitary operators are nontrivial in the case of quantum switch acting on the channels. We will discuss this fact in later sections. A general qubit unital channel has nine parameters (without unitary equivalence) which is slightly difficult to handle. Therefore, for this subsection, we will restrict ourselves to Pauli Channels. The subset consisting of the Pauli entanglement breaking channels satisfy $\sum_i|\lambda_i|\le 1$ and they form an octahedron in the $\lambda_i$ space: position of a channel, characterized by the set of three parameters $(\lambda_1,\lambda_2,\lambda_3)$, in the three dimensional Euclidean space is given by the coordinates $({\lambda}_1, {\lambda}_2, {\lambda}_2)$. The vertices of the octahedron are: $(+ 1, 0, 0)$, $(- 1, 0, 0)$,  $(0, + 1, 0)$, $(0, - 1, 0)$, $(0, 0, + 1)$, and $(0, 0, - 1)$. The figure \ref{cpuseful} shows the set of all those entanglement breaking Pauli channels that constitute the octahedron. The black region inside the octahedron corresponds to the set of entanglement breaking channels, the $\bar{C}_+$ branch of which -- after the action of the switch -- are being mapped outside the octahedron (shown in the blue color). The three vertices marked as $a,b$ and $c$ correspond to the points 
$(-1, 0, 0)$,$(0, -1, 0)$, $(0, 0, -1)$) respectively. The vertex $(0,0,-1)$ corresponds to the channel $(\sigma_x\rho\sigma_x+\sigma_y\rho\sigma_y)/2$ and the other two vertices are the unitary equivalents of this channel. These channels are mapped (under the action of the $\bar{C}_+$ branch of the quantum switch) onto the identity channel (corresponding to the point $(+ 1, + 1, + 1)$).
\begin{figure}[hbt!]
\centering
\begin{subfigure}{0.5\textwidth}
\centering
    \includegraphics[width=0.9\linewidth]{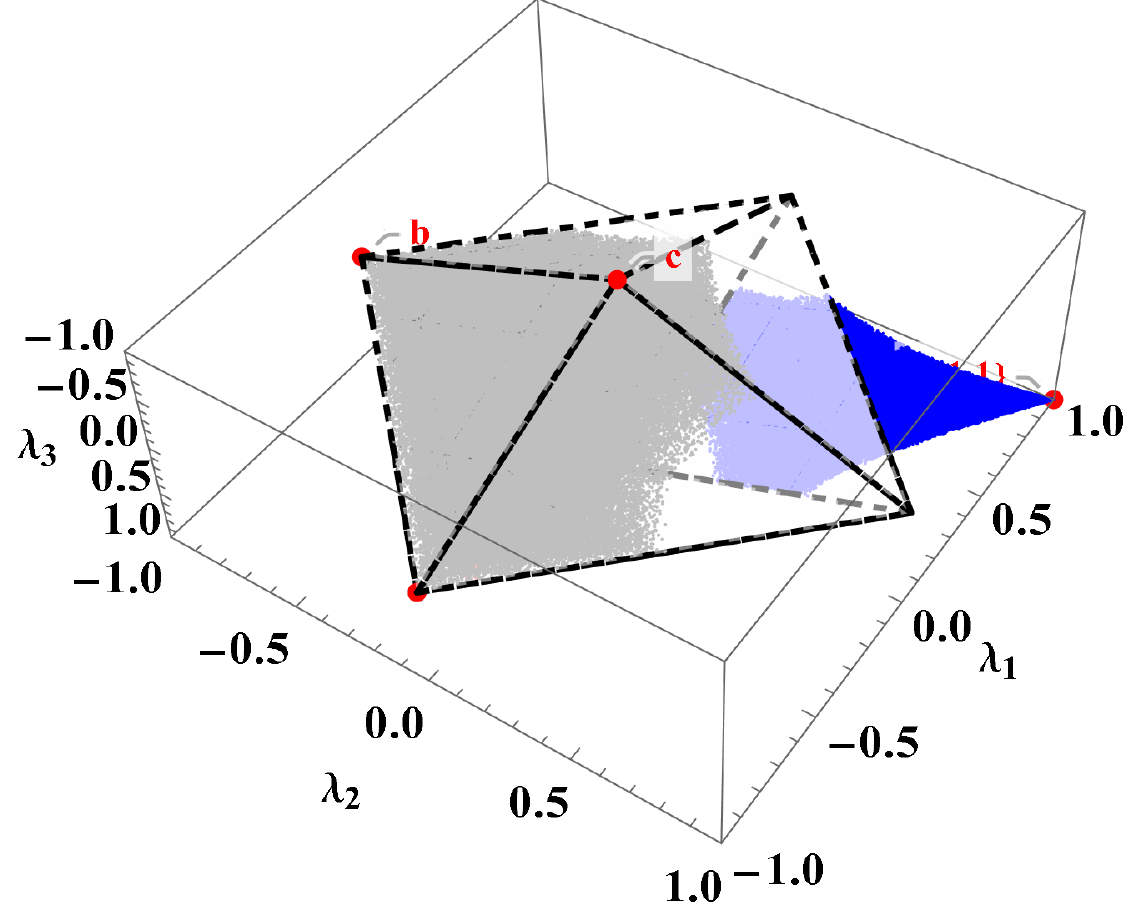}
    \caption{}
    \label{fig:Pauli-oct-c+.1}
\end{subfigure}%
\\
\begin{subfigure}{0.5\textwidth}
\centering
    \includegraphics[width=1.0\linewidth]{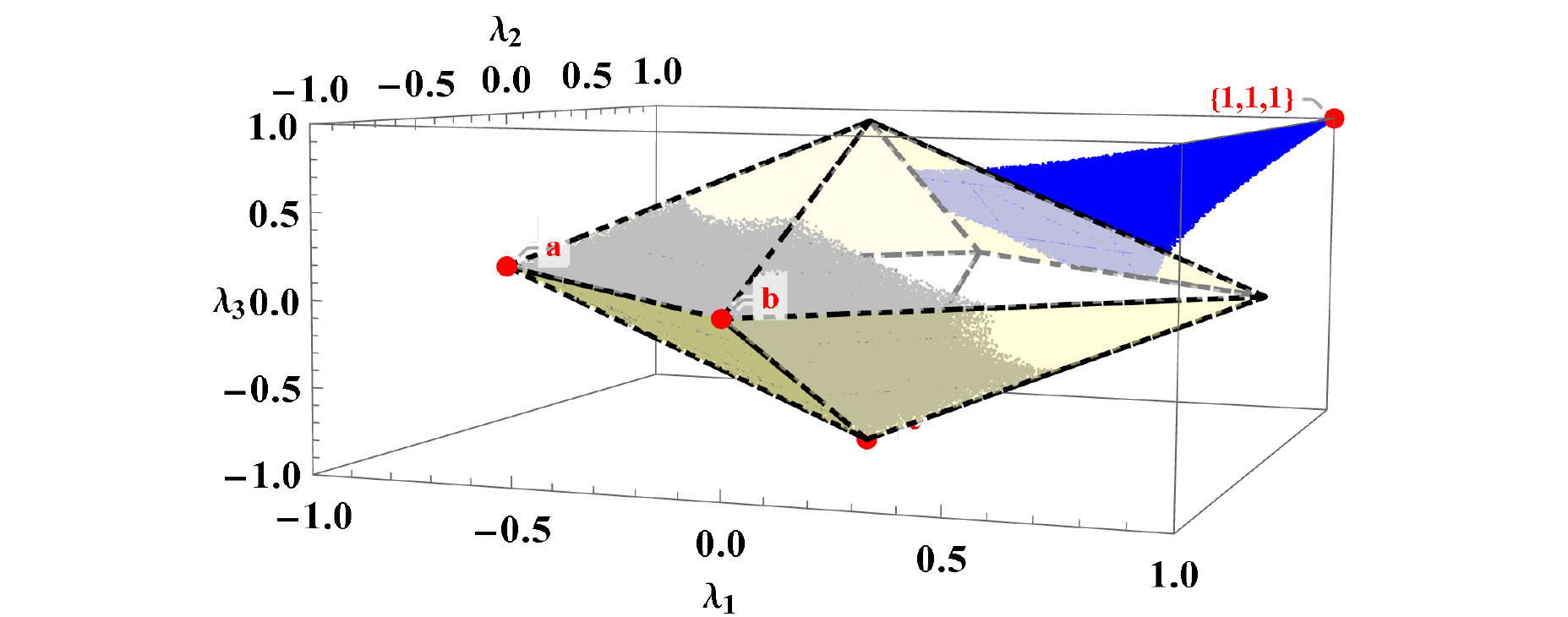}
    \caption{}
    \label{fig:Pauli-oct-c+.2}
\end{subfigure}
\caption{The 3-D plots show the octahedron of the entanglement breaking channels from two perspectives. The channels whose parameters $(\lambda_1,\lambda_2,\lambda_3)$ lie within the octahedron are entanglement breaking. The grey region inside the octahedron are the EBCs which, under the action of the $\bar{C}_+$, branch become non-EBC. The blue region is the mapping of these EBCs outside the octahedron. The points $a$, $b$, $c$ are the vertices of the octahedron given by $(-1, 0, 0)$,$(0, -1, 0)$, $(0, 0, -1)$ respectively. These channels are mapped onto the point $(1,1,1)$ corresponding to the identity channel.}
\label{cpuseful}
\end{figure}

\begin{figure}[hbt!]
\centering
\begin{subfigure}{0.5\textwidth}
\centering
    \includegraphics[width=0.85\linewidth]{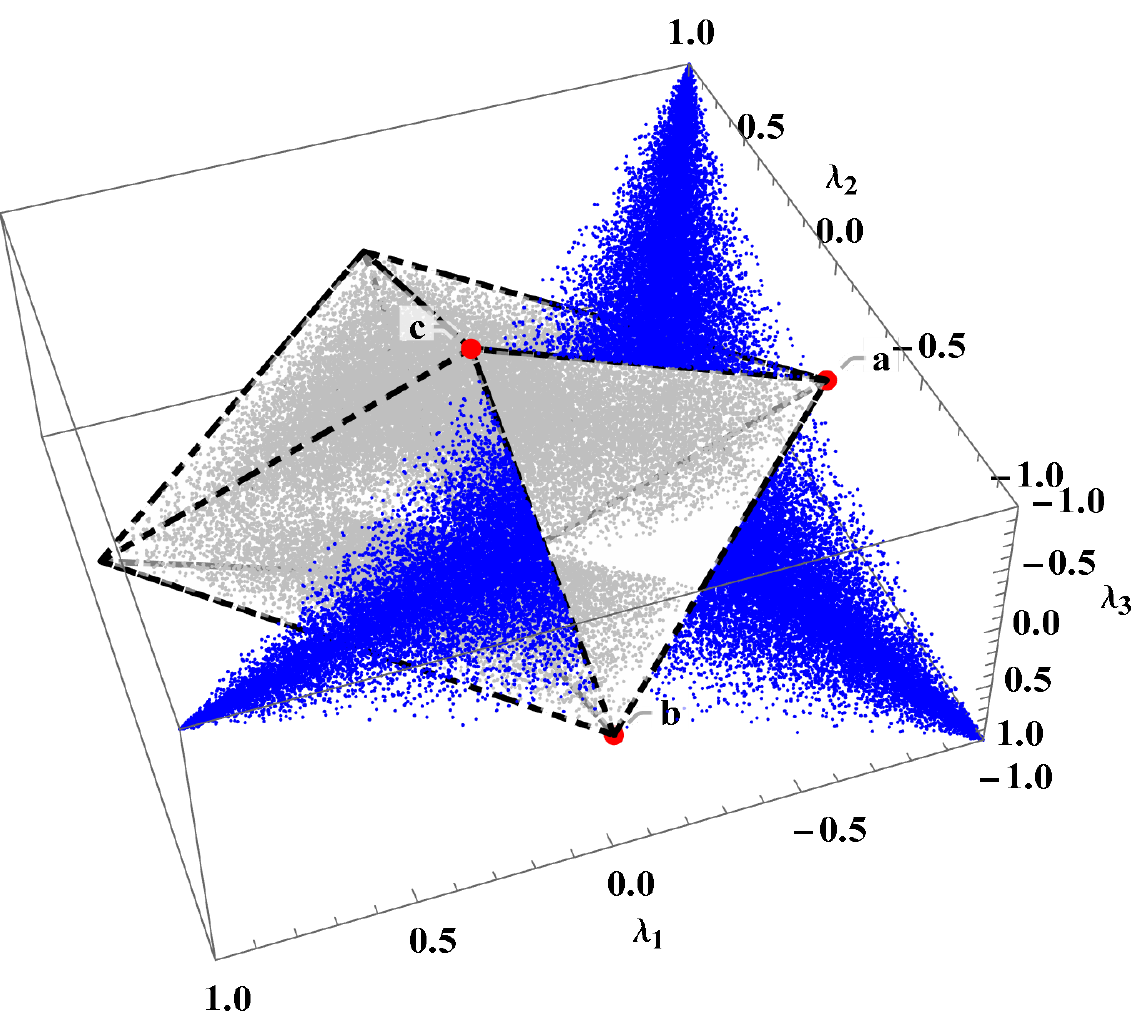}
    \caption{}
    \label{fig:Pauli-oct-c-.1}
\end{subfigure}%
\\
\begin{subfigure}{0.5\textwidth}
\centering
    \includegraphics[width=0.9\linewidth]{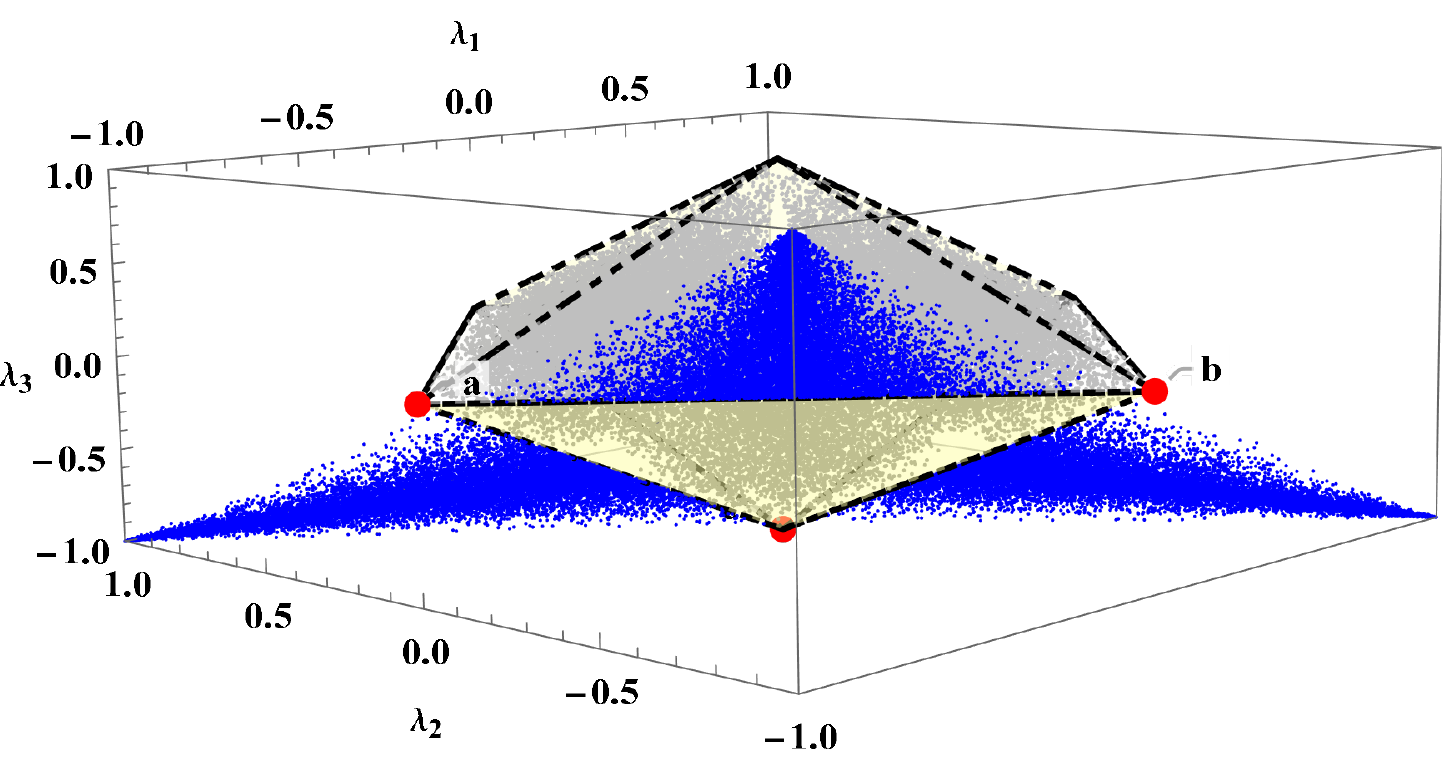}
    \caption{}
    \label{fig:Pauli-oct-c-.2}
\end{subfigure}
\caption{The figure shows (in grey) the EBCs inside the octahedron which are mapped to non-EBCs outside the octahedron under the $\bar{C}_-$ branch. These channels are mapped onto the blue region outside the octahedron which in this case is in a plane of the points $(1,-1,-1)$, $(-1,1,-1)$ and $(1,1,-1)$. These channels are $\sigma_Z\rho\sigma_Z$, $\sigma_Y\rho\sigma_Y$ and $\sigma_X\rho\sigma_X$ respectively. The points $a$, $b$, $c$ are the vertices of the octahedron given by $(-1, 0, 0)$, $(0, -1, 0)$, $(0, 0, -1)$ respectively and are mapped onto the channels $(1,-1,-1)$, $(-1,1,-1)$ and $(1,1,-1)$ respectively.}
\label{cmuseful}
\end{figure}
Figure \ref{cmuseful} shows the EBCs (in grey) that under the $\bar{C}_-$ branch are mapped outside the octahedron (turned into non-EBCs). The mapped channels outside the octahedron are marked in blue. The channels are mapped onto the plane containing the vertices $( 1, - 1, - 1)$, $(- 1, + 1, - 1)$ and $(- 1, - 1,  1)$ excluding the region that overlaps with one of the faces of the octahedron. The previously mentioned points $a$, $b$ and $c$ of the octahedron are now matched to exactly these vertices $( 1, - 1, - 1)$, $(- 1, + 1, - 1)$, $(- 1, - 1,  1)$ respectively under the $\bar{C}_-$ branch. These three points $( 1, - 1, - 1)$, $(- 1, + 1, - 1)$, $(- 1, - 1,  1)$ correspond to the channels $\rho\rightarrow \sigma
_X\rho\sigma_X$, $\sigma_Y\rho\sigma_Y$ and $\sigma_Z\rho\sigma_Z$ respectively and are equivalent to the identity channel upto a unitary transformation. These vertices along with the identity channels $\{1,1,1\}$, form the extremal points of the tetrahedron that contains all qubit channels.

\subsection{Concatenation of quantum channels and the quantum Switch}\label{subsec:conc_ch}
It is well known that if $\Lambda$ is EBC (or IBC) then $\Phi\circ\Lambda$ is also EBC (or IBC) for any channel $\Phi$. Therefore, if EBC (or IBC) is unavoidable for communication, concatenating it straightforwardly with another channel, i.e. without an external resource like superposition of direct pure processes or SDPP \cite{PRA-Rubino}, will not help in the transfer of entanglement (or steerable states).\\
Now suppose a quantum channel $\Lambda$ is useless even after the use of the quantum switch. Then we show below that there may exist another quantum channel $\Phi$ so that $\Phi\circ\Lambda$ is useful under the action of quantum switch (The operation is demonstrated in Fig. \ref{fig:conc-useful}). Therefore, the concatenation of quantum channel may provide an advantage when the communication is supported by the quantum switch. One should note that this concatenation is \textit{before} the implementation of the switch (Fig. \ref{fig:conc-useful}) and is different from the concatenation that happens intrinsically in the circuit that implements the switch (Fig. \ref{fig:switch}). It has been shown in previous studies \cite{PRA-Rubino} that such intrinsic concatenation when used in an SDPP (without indefinite causal order), can also provide an advantage in terms of quantum capacity. But the aforesaid study of \cite{PRA-Rubino} is quite different than what we study in this section. Now at first, we make the following two observations, and then we have a detailed study for the case of Pauli channels and three-parameter non-unital channel.

\begin{figure}[hbt!]
\includegraphics[width=7cm,height=5cm]{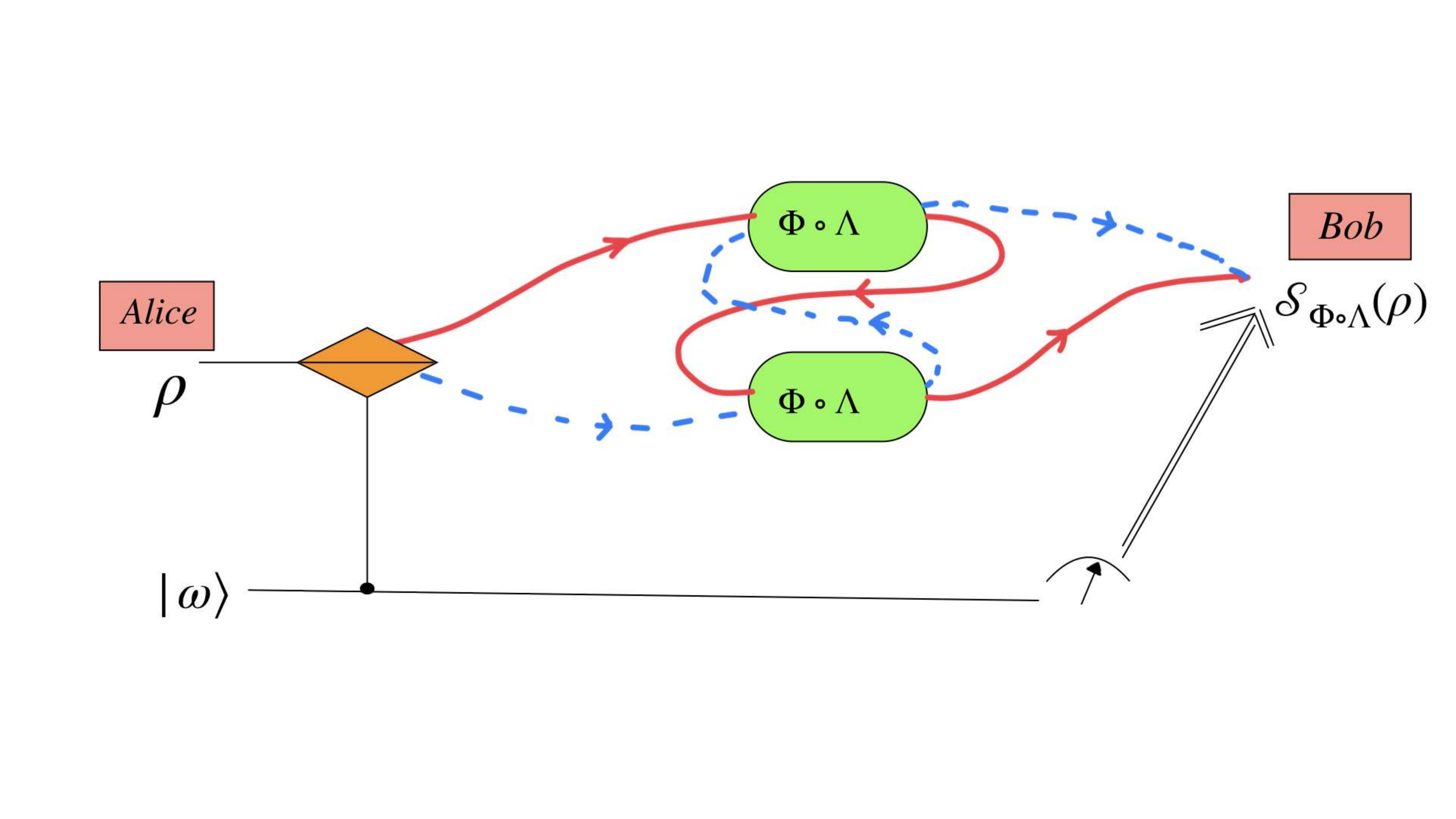}
\caption{The figure shows the concatenation of quantum channels before the implementation of the switch. In this respect this figure is different from the figure \ref{fig:switch}. In this section we compare the effect of quantum switch on the channel $\Lambda$ i.e., the property of ($S_{\Lambda,\Lambda,\omega}$) with the effect of quantum switch on the channel $\Phi\circ\Lambda$ i.e., the property of $S_{\Phi\circ\Lambda,\Phi\circ\Lambda,\omega}$ (represented as $\mathcal{S}_{\Phi\circ\Lambda}$ in the figure) as shown above.}
\label{fig:conc-useful}
\end{figure}
\begin{obs}
There exist quantum channels which do not provide any advantage under the action of quantum switch, but such a channel may become useful under the action of quantum switch if it is concatenated with another quantum channel.  \label{conc-switch}
\end{obs}

\begin{proof}
Suppose, $\Lambda(\rho)=\frac{1}{2}\rho+\frac{1}{2}\sigma_z\rho\sigma_z$. Now, from the Theorem 4 of \citep{Rus03}, it is clear that $\Lambda$ is an EBC.  From the equation \eqref{1st_branch_random} and equation \eqref{2nd_branch_random} we obtain that

\begin{equation}
q\bar{C}_{\Lambda,+}(\rho)=\frac{1}{2}\rho+\frac{1}{2}\sigma_z\rho\sigma_z=\Lambda(\rho);~(1-q)\bar{C}_{\Lambda,-}(\rho)=0.
\label{no_advantage_chan}
\end{equation} 

Therefore, since $\Lambda$ is an EBC, from equation \eqref{no_advantage_chan}, we get that this channel can not give advantage under the action of quantum switch.

Now suppose $\Lambda^{\prime}=\Phi\circ\Lambda$ where $\Phi(\rho)=\sigma_x\rho\sigma_x$. Then,

\begin{equation}
\Lambda^{\prime}(\rho)=\frac{1}{2}\sigma_x\rho\sigma_x+\frac{1}{2}\sigma_y\rho\sigma_y.
\end{equation}

This is the channel used in Ref. \cite{chiribella_perfect}. Therefore, we know that perfect communication can be achieved through this quantum channel using the quantum switch. It can be easily shown there exist several other such examples of channels.  
\end{proof}
This example also indicates that the action of unitary on the output space has a non-trivial effect when the action of quantum switch is considered. Since $T$-matrices of the Pauli channels are diagonal, those channels are commuting, and therefore, we conclude that the action of a unitary on the input space also has non-trivial effects when the action of the quantum switch is considered.

The results of the concatenation are more interesting if $\Phi$ is also an EBC (i.e., a useless channel) as we show in the next observation:

\begin{obs}
There exist quantum channels which do not provide any advantage under the action of quantum switch, but such a channel may become useful under the action of quantum switch if it is concatenated with another EBC. 
\end{obs}

\begin{proof}
 Consider the quantum channels
\begin{equation}
\Lambda(\rho)=\frac{1}{2}\rho+\frac{1}{2}\sigma_z\rho\sigma_z \label{bad_chan}
\end{equation} 
 
 and 
 
 \begin{equation}
 \Phi(\rho)=\frac{1-\lambda_3}{4}\rho+\frac{1+\lambda_3}{4}\sigma_x\rho\sigma_x+\frac{1+\lambda_3}{4}\sigma_y\rho\sigma_y+\frac{1-\lambda_3}{4}\sigma_z\rho\sigma_z\label{good_chan}
 \end{equation}
 
where the $-1\leq \lambda_3\leq 1$. 
 
 From equation \eqref{no_advantage_chan}, we know that $\Lambda(\rho)$ can not give advantage using quantum switch.
  
 Now,
 \begin{align}
 \Lambda^{\prime}(\rho)&=(\Phi\circ\Lambda)(\rho)\nonumber\\
 &=\Phi(\rho).
 \end{align}

Note that $\Phi$ has $T$-matrix

\begin{align}
\mathcal{T}_{\Phi} = 
\begin{pmatrix}
1 & 0 & 0 & 0\\
0 & 0& 0& 0\\
 0 & 0 & 0 &0\\
0&0&0& -\lambda_3\\
\end{pmatrix}
\end{align}

Therefore, from the Theorem 4 of \citep{Rus03}, we get that $\Phi$ is an EBC for $-1\leq\lambda_3\leq 1$.

Now from the equation \eqref{1st_branch_random} and equation \eqref{2nd_branch_random} we obtain that

\begin{align}
\bar{C}_{\Phi,+}=&[\frac{(1+\lambda_3^2)}{4}\rho+\frac{(1-\lambda_3^2)}{8}\sigma_x\rho\sigma_x
\nonumber\\
&+\frac{(1-\lambda_3^2)}{8}\sigma_y\rho\sigma_y+\frac{(1-\lambda_3)^2}{8}\sigma_z\rho\sigma_z]/q
\end{align}

and

\begin{align}
\bar{C}_{\Phi,-}=&[\frac{(1-\lambda_3^2)}{8}\sigma_x\rho\sigma_x+\frac{(1-\lambda_3^2)}{8}\sigma_y\rho\sigma_y\nonumber\\
&+\frac{(1+\lambda_3)^2}{8}\sigma_z\rho\sigma_z]/(1-q)
\end{align}
 
with $q=2\lambda_3+3\lambda_3^2$. 

Now if we implement a measurement-based controlled operation  $\cU=\Id\otimes \ket{+}\bra{+}+\sigma_z\otimes\ket{-}\bra{-}$ and we prepare the control state as $\omega=\ket{+}\bra{+}$, then the effective channel on the system (after tracing out the control part) is
 \begin{align}
 \cC_{eff}(\rho)=&\tr_{\cH_c}[\cU S_{\Phi,\Phi,\omega}(\rho)\cU^{\dagger}]\nonumber\\
=&\dfrac{3+2\lambda_3+3\lambda^2}{8}\rho+\dfrac{(1-\lambda_3)^2}{8}\sigma_z\rho\sigma_z \nonumber\\& +\dfrac{1-\lambda_3^2}{4}(\sigma_x\rho\sigma_x+\sigma_y\rho\sigma_y)\label{good_chan_switched}
\end{align}
The T-matrix of the above channel $\cC_{eff}$ is given by
\begin{align}
\mathcal{T}_{\cC_{eff}} = 
\begin{pmatrix}
1 & 0 & 0 & 0\\
0 & \frac{(1+\lambda_3)^2}{4}& 0&0\\
 0 & 0 & \frac{(1+\lambda_3)^2}{4}&0\\
0&0&0&  \lambda_3^2\\
\end{pmatrix}
\end{align}
Therefore, from the Theorem 4 of \citep{Rus03}, for $\lambda_3^2+\frac{(1+\lambda_3)^2}{2}> 1$  or equivalently for $1 \geq\lambda_3> \frac{1}{3}$, the channel $\cC_{eff}$ is not EBC.   It can be easily shown there exist several other such examples of channels.
\end{proof}
The above result can be explored further by scanning through the space of all Pauli channels and one can quantify the advantage that concatenation can provide with the use of the switch.  Let us first take a general Pauli channel $\Lambda$ given by the following $T$ matrix
\begin{align}
\mathcal{T}_{\Lambda} = 
\begin{pmatrix}
1 & 0 & 0 & 0\\
0 & \lambda_1& 0&0\\
 0 & 0 & \lambda_2&0\\
0&0&0&  \lambda_3\\
\end{pmatrix}
\end{align}
We generate a random pair of Pauli EBCs $\mathcal{T}_{\Lambda}$ and $\mathcal{T}_{\Gamma}$ that are not useful under either or both branches of the switch and check whether the concatenated channel $\mathcal{T}_{\Lambda \circ\Gamma}$ is useful under the $C_+$ or $C_-$  branch. We can also compare the results of a concatenation of Pauli EBCs with the concatenation of three parameters non-unital EBCs (say $\Lambda^{\prime}$) given by
\begin{align}
\mathcal{T}_{\Lambda^{\prime}} = 
\begin{pmatrix}
1 & 0 & 0 & 0\\
0 & k_1& 0&0\\
 0 & 0 &k_1&0\\
t&0&0& k_3\\
\end{pmatrix}.
\end{align}
 We present the results in Fig. \ref{venn}.
\begin{figure}[hbt!]
\centering
\begin{subfigure}{0.5\textwidth}
\centering
    \includegraphics[width=0.75\linewidth]{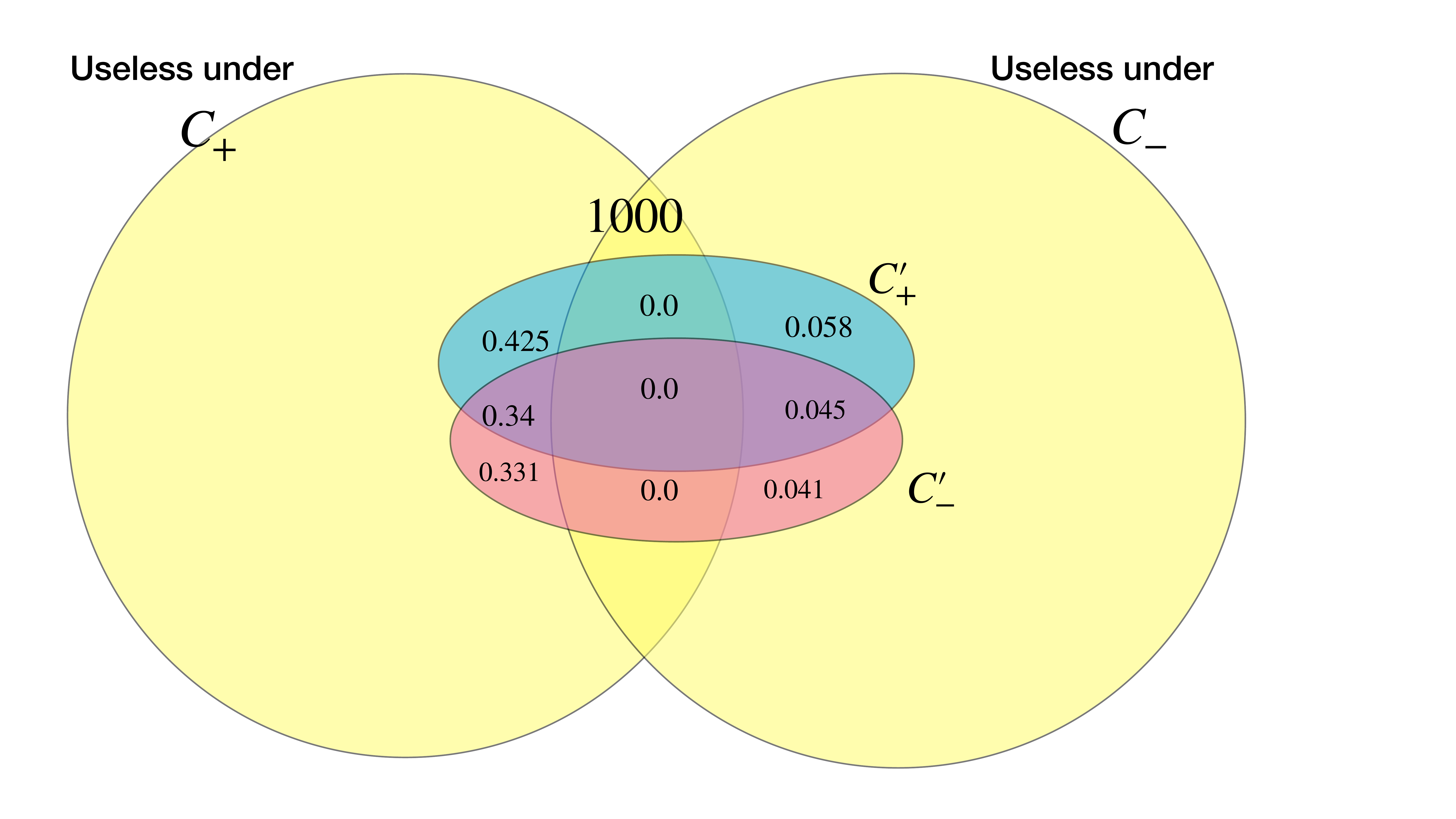}
    \caption{Pauli Channels}
    \label{venn1}
\end{subfigure}%
\\
\begin{subfigure}{0.5\textwidth}
\centering
    \includegraphics[width=0.75\linewidth]{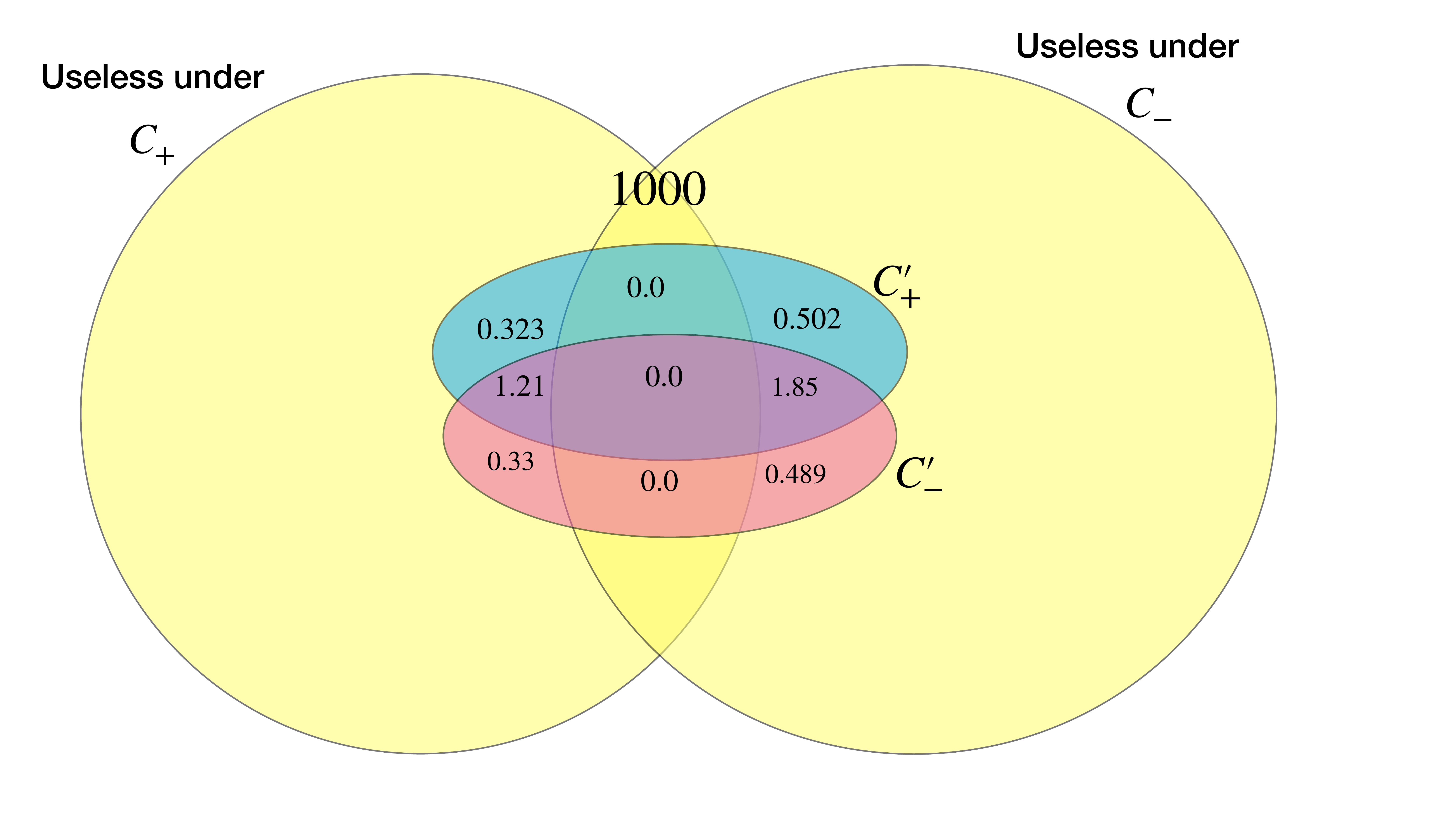}
    \caption{Non-unital Channels}
    \label{venn2}
\end{subfigure}
\caption{The Venn diagrams show the fraction of useless channel pairs that become useful after concatenation. The larger circles (in yellow) contain the set of EBCs that remain entanglement breaking under either $C_+$ or $C_-$ or both. ]Recall that in the case of non-unital maps, the branches $C_{\pm}$ are completely positive maps but not trace preserving. We concatenate two channels that are both useless under $C_+$ (or $C_-$ or both). The sets denoted by $C_+'$ and $C_-'$ contain the pairs that after such a concatenation become non-entanglement breaking under either $C_+$ or $C_-$ branches respectively. Note that the fractions mentioned in the diagram are out of a 1000 pairs of channels. For instance, from the figure \ref{venn1} we see that out of $10^6$ (randomly chosen) useless pairs of Pauli channels (useless under either $C_+$ or $C_-$ or both), $425$ pairs of channels are those pairs that are useless only under $C_+$ and on concatenation become useful only under $C_+$. To estimate these measures we generated about $10^7$ random pairs of EBCs and check for their behaviour under the switch after concatenation.}
\label{venn}
\end{figure}
A useless channel is completely useless if under both $C_+$ and $C_-$ it remains useless. Such channels remain useless \emph{even after} the use of the quantum switch. From Fig. \ref{venn} it can be found that such channels remain completely useless even after concatenation- both for the Pauli channels and the particular class of non-unital channels that we have used here. This fact encourages us to provide the following conjecture:\\ 

\textbf{Conjecture 1.} \emph{Concatenation of two completely useless channels is always completely useless.}\\ 

A particularly useful set of channels is at the intersection of $C_+'$ and $C_-'$in Fig. \ref{venn}. From this figure, we see that there exist channels that only probabilistically can transfer entanglement  under the switch, but after concatenation with another channel can transfer entanglement deterministically under the switch. 
\\
As an example of the effect of concatenation, figure \ref{cpconcat} shows some of the Pauli channels useless under $C_+$ inside the octahedron, which under concatenation become useful. The orange dots show the map of the concatenated channels (in red)  outside of the octahedron under the $C_+$ channel.
\begin{figure}[hbt!]
\includegraphics[scale=.6]{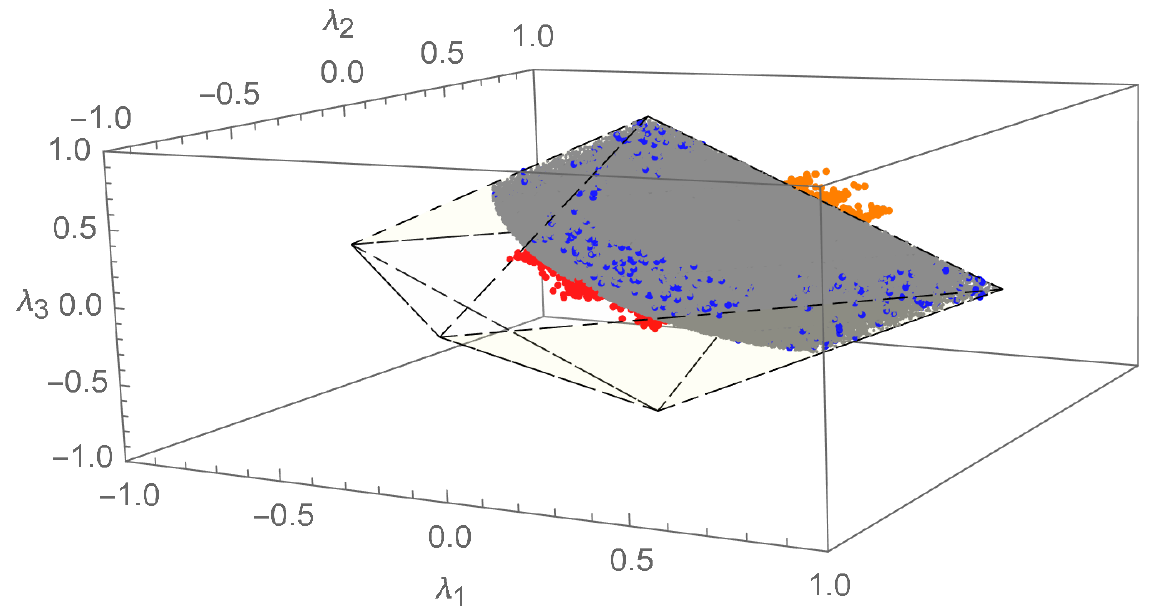}
\caption{The grey region shows the useless Pauli channels under concatenation. The blue dots are the useless channels that can be converted to useful under concatenation for $C_+$ channel. The red dots are the result of the concatenations, and the orange dots are the mapping of the red dots outside the octahedron under $C_+$ channel.}
\label{cpconcat}
\end{figure}
We also see that, there is an interesting geometry of the concatenations that results in useful channels under $C_+$ or $C_-$. Figure \ref{dist} shows the Euclidean distances between the Pauli channels that become useful under one of the branches of the switch after concatenation. The Euclidean distance between the channels is calculated as $D(\Lambda,\Gamma)=\sqrt{\sum(\lambda_i-\gamma_i)^2}$ where $[\cT_{\Gamma}]_{ii}=\gamma_i$ for all $i\in\{1,2,3\}$. We can see that the pairs that concatenate into useful channels under switch, are not located close to each other. In fact these distances are comparable to the edge length of the octahedron, i.e., $\sqrt{2}$.
\begin{figure}[hbt!]
\centering
\begin{subfigure}{0.5\textwidth}
\centering
\includegraphics[width=0.75\linewidth]{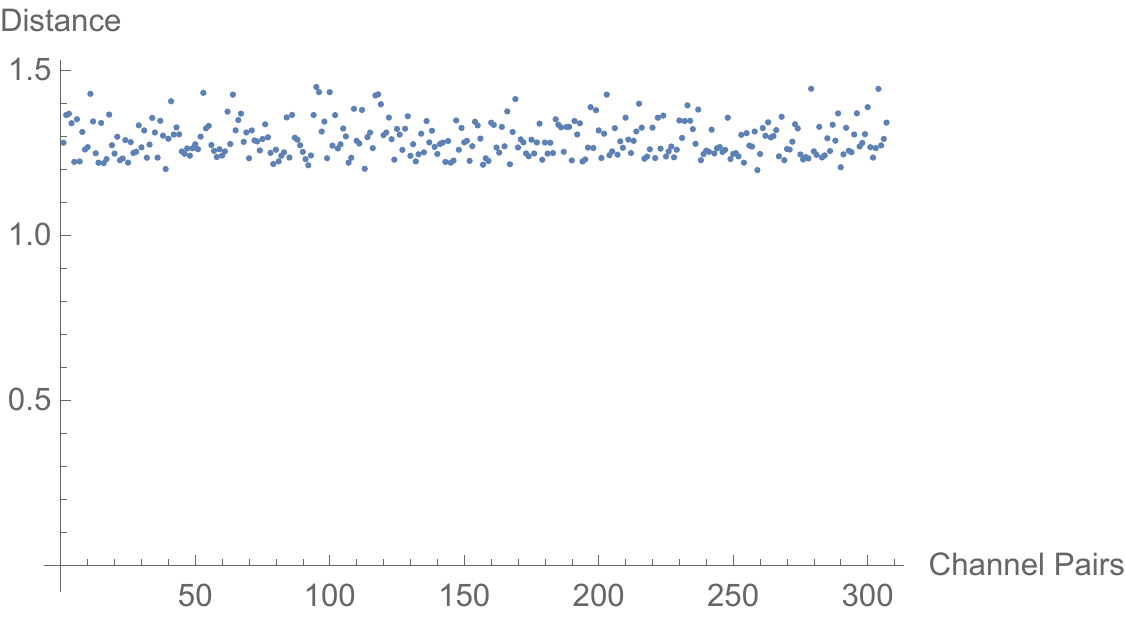}
    \caption{}
    \label{dist1}
\end{subfigure}%
\\
\begin{subfigure}{0.5\textwidth}
\centering
    \includegraphics[width=0.75\linewidth]{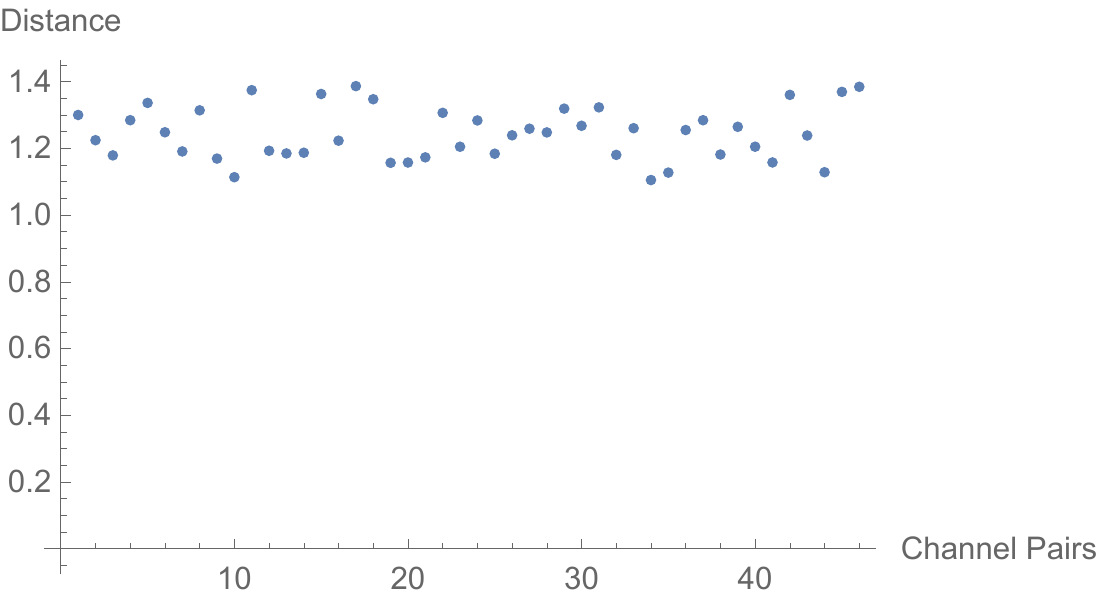}
    \caption{}
    \label{dist2}
\end{subfigure}
\caption{\ref{dist1}: Spread of the euclidean distance between the Pauli channels that are useless under $\bar{C}_+$ and which, after concatenation, give useful channels under $\bar{C}_+$. We have indexed the set of channel pairs with integers which are plotted along the x-axis. The average distance is $\sim$ 1.28 \ref{dist2}: Spread of the euclidean distance between the channels which are useless under $\bar{C}_-$ and which after concatenation give useful channels under $\bar{C}_-$. The average distance for such pairs is $\sim$ 1.21, which is comparable to the length of the edge of the octahedron.
}
\label{dist}
\end{figure}
\subsection{Advantages in different information theoretic tasks}\label{subsec:Advantages_information}
\subsubsection{Advantage in quantum random access code}
\label{subsubsec:adv-QRAC}
Let, Alice has $n$-dit string denoted by $\vec{x}=(x_1,.....,x_n)$ at her disposal. She encodes a particular $n$-dit string in a state of a $d$-level quantum system (i.e., in a qudit) and then transfers this qudit to Bob. In addition, Bob receives a random number $j \in \{1, 2, \ldots, n\}$. Now, Bob's task is to guess the $j$-th dit $x_j$. He does this by doing a measurement on the quantum system sent by Alice. He has $n$ choices of measurements each with $d$ outcomes $0, 1, \ldots, d - 1$. After obtaining the random number $j$, he performs the $j$-th measurement on the qudit sent by Alice. Depending on the outcome of the measurement, Bob guesses the $j$-th entry $x_j$ of Alice's string $\vec{x}$. Let, his guess be $y$. The game will be successful if $y=x_j$. This is known as $(n,d)-$quantum random access code (QRAC). It has a classical counterpart, known as $(n,d)-$random access code (RAC), where Alice is allowed to send dits to Bob instead of qudits. The maximum average success probability of random access code (for $n = 2$ case) is $P^{(2,d),max}_{rac}=\frac{1}{2}(1+\frac{1}{d})$. But the maximum average success probability of quantum random access code is $P^{(2,d),max}_{qrac}=\frac{1}{2}(1+\frac{1}{\sqrt{d}})$. Therefore, $P^{(2,d),max}_{rac}< P^{(2,d),max}_{qrac}$. We will call a particular encoding scheme of Alice and a particular set of measurements performed by Bob to guess the desired $d$-level entry in Alice's $n$-dit string as a ``useful strategy" if it can achieve a quantum advantage in the average success probability i.e., the average success probability $P^{(2,d)}_{qrac}>P^{(2,d),max}_{rac}$. It is shown in \cite{Heino_random} that the incompatibility of measurement is necessary for $(2,d)-$ QRAC  to be useful.
\\
Now, suppose only a noisy channel $\Lambda$ is available to Alice to transfer the qudit to Bob. In this case, the achievable success probability decreases (assuming that the quantum switch is unavailable to Alice and Bob). Depending on $\Lambda$, this decrement of probability can be drastic. In this context, we have the following theorem:

\begin{theorem}\label{Th:2-IBC,QRAC}
If Alice has only $2-IBC$ channels to transfer qudits to Bob in a $(2,d)-QRAC$ game, there exist no measurement strategy for Bob as well as qudit encoding scheme for Alice which can be useful i.e., using which one can get quantum advantage.
\end{theorem}

\begin{proof}
Let Alice encodes a $2-dit$ string $\vec{x} = (x_1, x_2)$ in the qudit state $\cE(\vec{x})$ and sends it to Bob through channel $\Lambda$. Now suppose Bob chooses the measurement $M_1=\{M_1(j)\}_{j\in\{ 0,1,\ldots,d-1\}}$ for the first dit and $M_2=\{M_2(j)\}_{j\in\{ 0,1,\ldots,d-1\}}$ for the second dit. Then the average success probability is

\begin{align}
P^{(2,d)}_{qrac}&=\frac{1}{2d^2}\sum_{\vec{x}}\tr[\Lambda(\cE(\vec{x})) (M_1(x_1)+M_2(x_2))]\nonumber\\
&=\frac{1}{2d^2}\sum_{\vec{x}}\tr[\cE (\vec{x})\Lambda^*(M_1(x_1)+M_2(x_2))]\nonumber\\
&\leq \frac{1}{2d^2}\sum_{\vec{x}}\mid\mid \Lambda^*(M_1(x_1))+\Lambda^*(M_2(x_2))\mid\mid\nonumber\\
&\leq \frac{1}{2}(1+\frac{1}{d})
\end{align}

where we get the last inequality from Theorem $1$ of \cite{Heino_random} and noting the fact that $\Lambda^*(M_1)$ and $\Lambda^*(M_2)$ are compatible.
\end{proof}

Therefore, $2-IBC$ channels are useless for $(2,d)-$QRAC \cite{Heino_random}. Note that incompatibility is necessary but not sufficient for $(2,d)-$QRAC. The analysis of previous sections indicates that a quantum switch can help to get a quantum advantage in $(2,d)-$QRAC though the $\Lambda$ is $2-IBC$ in some cases. We will study this through the next example.
 
\begin{example}\label{ex_prob_adv_qrac}
\rm{Suppose in a $(2,2)$-QRAC, Alice is encoding the two bit string ($ij$) in the quantum state $\rho(ij)=\ket{+,\hat{n}_{ij}}\bra{+,\hat{n}_{ij}}$ where $\ket{+,\hat{n}_{ij}}$ is the eigen state of $\hat{n}_{ij}.\vec{\sigma}$ corresponding to the eigen value $+1$ and $\hat{n}_{ij}=\frac{1}{\sqrt{2}}((-1)^{i},(-1)^{j},0)^T$ for $i,j\in\{0,1\}$. In this case, if Alice can transfer the state to Bob with perfect communication i.e., through the identity channel and Bob chooses the measurement $M_1=\{\ket{+,\hat{x}}\bra{+,\hat{x}},\ket{-,\hat{x}}\bra{-,\hat{x}}\}$ for the first bit and $M_2=\{\ket{+,\hat{y}}\bra{+,\hat{y}},\ket{-,\hat{y}}\bra{-,\hat{y}}\}$ for the second bit ($\ket{\pm,\hat{j}}$ is the eigen state of $\sigma_j$ with eigen value $\pm 1$ for $j = x, y$), then the average success probability is
\\
\begin{align}
P^{(2,2)}_{qrac}&=\frac{1}{8}\sum_{ij}\tr[\rho(ij)(M_1(i)+M_2(j))]\nonumber\\
&=\frac{1}{2}\left(1+\frac{1}{\sqrt{2}}\right)
\end{align}
which is the optimal average success probability for $(2,2)-$QRAC.\\
Now, suppose Alice has only $2-IBC$s to communicate with Bob. For example, suppose Alice has the channel $\Phi$ given in equation \eqref{good_chan}. Then, from Theorem \ref{Th:2-IBC,QRAC}, we know that no strategy of Alice and Bob can be useful. Now after using quantum switch, the effective channel $\cC_{eff}(\rho)=Tr_{\cH_c}[\cU S_{\Phi,\Phi,\omega}(\rho)\cU^{\dagger}]$ is given in equation \eqref{good_chan_switched}. If Alice uses this effective channel, the average success probability is 
\\
\begin{align}
P^{\prime(2,2)}_{qrac}&=\frac{1}{8}\sum_{ij}\tr[\cC_{eff}(\rho(ij))(M_1(i)+M_2(j))]\nonumber\\
&=\frac{(1+\lambda_3)^2}{4}P^{(2,2)}_{qrac}+\frac{(1-\frac{(1+\lambda_3)^2}{4})}{2}\nonumber\\
&=\frac{(1+\lambda_3)^2}{8\sqrt{2}}+\frac{1}{2}\nonumber\\
&=\frac{1}{2}[1+\frac{(1+\lambda_3)^2}{4\sqrt{2}}]
\end{align}
\\
Since,  $P^{(2,2),max}_{rac}=\frac{3}{4}$, we have $P^{\prime(2,2)}_{qrac}>P^{(2,2),max}_{rac}$ only for $\lambda_3>2^{\frac{3}{4}}-1\approx 0.6818$. Therefore, clearly $\cC_{eff}$ is not 2-IBC for $1\geq\lambda_3 > 2^{\frac{3}{4}}-1$. In FIG. \ref{fig:prob_random}, variation of $P^{\prime(2,2)}_{qrac}$ w.r.t. $\lambda_3$ is plotted.

\begin{figure}[hbt!]
\centering
\includegraphics[scale=0.48]{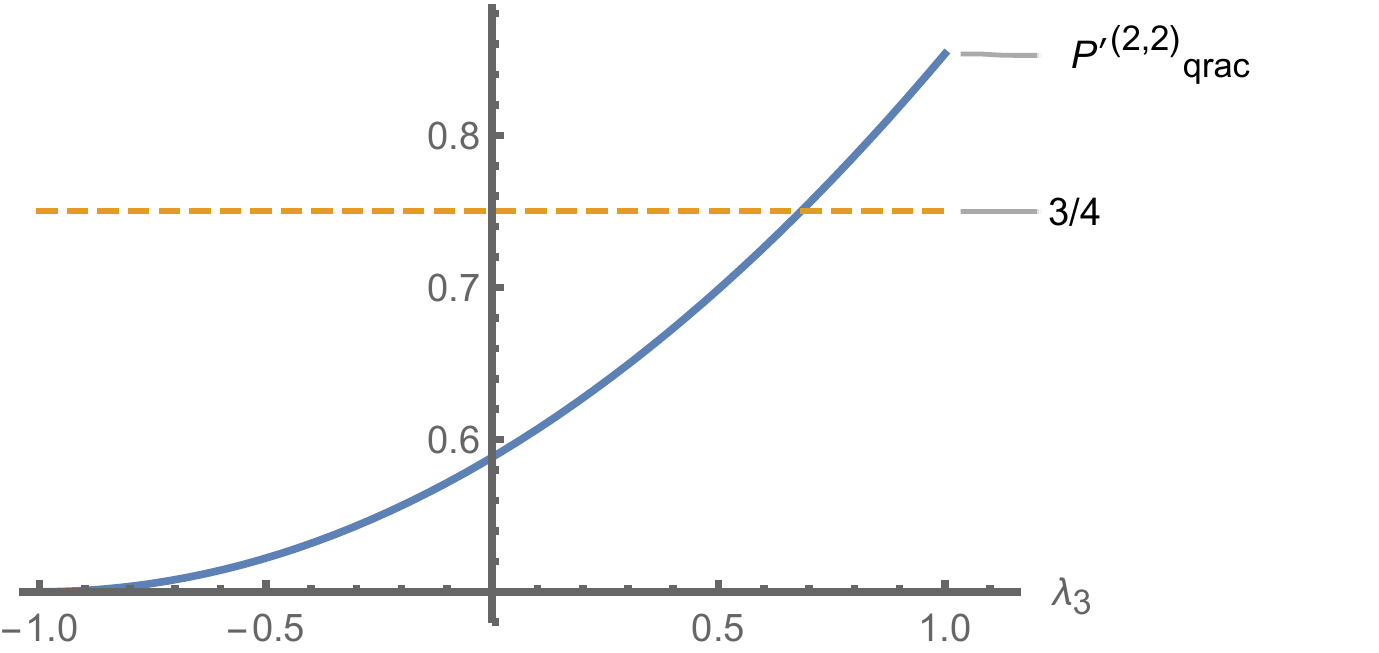}
\caption{Average success probability $P^{\prime(2,2)}_{qrac}$ vs. $\lambda_3$ graph. This graph shows that $P^{\prime(2,2)}_{qrac}>P^{(2,2),max}_{rac}=\frac{3}{4}$ for $1\geq\lambda_3 > 2^{\frac{3}{4}}-1$. }
\label{fig:prob_random}
\end{figure}
}
\end{example}

\subsubsection{Advantage in quantum steering}\label{subsubsec:adv-steer}
Suppose Alice and Bob share a bipartite quantum state $\rho^{AB}\in\cS(\cH_A\otimes\cH_B)$ and Alice has a measurement assemblage (i.e., a set of measurement) $\cM_A=\{M_x\}$. Each $M_x$ has the outcome set $\Omega_x$. Now state $\rho^{AB}$ is called unsteerable from Alice to Bob i.e., from $A$ to $B$ with $\cM_A$ if there exist a probability distribution $\pi_{\lambda}$, a set of states $\{\sigma_{\lambda}\}$  of $B$ and a set of probability distributions $P_A(a_x|x,\lambda)$ such that for all $a_x\in\Omega_x$ and for all $x$ the equality
\begin{equation}
\rho^B_{a_x|x}=\tr_A[(M_{a_x|x}\otimes\Id_B)\rho^{AB}]=\sum_{\lambda}\pi_{\lambda}P_A(a_x|x,\lambda)\sigma_{\lambda}
\end{equation}
holds. A state $\rho^{AB}$ is called unsteerable from $A$ to $B$ if it is unsteerable from $A$ to $B$ with any measurement assemblage $\cM_A$. A state $\rho^{AB}$ is called unsteerable if it is unsteerable from both $A$ to $B$ and from $B$ to $A$. Otherwise, it is called steerable. Steering is useful for different information-theoretic tasks \cite{steer-review}. There exist several steering inequalities each of which detects the steerability of a bipartite quantum state \cite{Das-steer-ineq, Wiseman-steer-ineq, Cao-steer-ineq, Costa-steer-ineq}. One such steering inequality using two measurement settings, we are writing below which was originally studied in \cite{Wiseman-steer-ineq, Das-steer-ineq, Costa-steer-ineq}. The bipartite state $\rho^{AB}$ is unsteerable then
\begin{equation}
F(\rho^{AB})=\frac{1}{\sqrt{2}}\mid\tr[\rho^{AB}(\sigma_x\otimes\sigma_x)]+\tr[\rho^{AB}(\sigma_z\otimes\sigma_z)]\mid\leq 1.\label{steer-ineq}
\end{equation}
\\
Now suppose Bob is preparing a bipartite entangled state $\rho^{AB}$ and sending the part $A$ to Alice. This shared bipartite entangled state, can be used in different information-theoretic tasks that can be performed with the help of $A$ to $B$ steering. We now have the following result in this connection (assuming that the quantum switch is not available to Alice and Bob)-

\begin{theorem}
Suppose Bob is sending the part $A$ of a bipartite state $\rho^{AB}$ to Alice (i.e., from $B$ to $A$) through an $n$-incompatibility breaking channel $\Lambda$ then $\rho^{\prime AB}=(\Lambda\otimes\cI)(\rho^{AB})$ is steerable from $A$ to $B$ with no measurement assemblage $\cM_A=\{M_x\}^n_{x=1}$ where $\cI$ is the identity channel. 
\end{theorem}

\begin{proof}
Consider an arbitrary measurement assemblage $\cM_A=\{M_x\}^n_{x=1}$ for Alice. As $\Lambda$ is a $n$-IBC, \{$\Lambda^*(M_x)$\} is a compatible set and hence, has a joint observable $\cJ=\{J_{\lambda}\}$ such that $\Lambda^*(M_{a_x|x})=\sum_{\lambda}P_A(a_x|x,\lambda)J_{\lambda}$ where $P_A(a_x|x,\lambda)$ is some probability distribution. Then

\begin{align}
\rho^{\prime B}_{a_x|x}&=\tr_A[(M_{a_x|x}\otimes\Id_B)\rho^{\prime AB}]\nonumber\\
&=\tr_A[(M_{a_x|x}\otimes\Id_B)(\Lambda\otimes\cI)(\rho^{AB})]\nonumber\\
&=\tr_A[(\Lambda^*(M_{a_x|x})\otimes\Id_B)\rho^{AB}]\nonumber\\
&=\tr_A[(\Lambda^*(M_{a_x|x})\otimes\Id_B)\rho^{AB}]\nonumber\\
&=\tr_A[(\sum_{\lambda}P_A(a_x|x,\lambda)J_{\lambda}\otimes\Id_B)\rho^{AB}]\nonumber\\
&=\sum_{\lambda}P_A(a_x|x,\lambda)\tr_A[(J_{\lambda}\otimes\Id_B)\rho^{AB}]\nonumber\\
&=\sum_{\lambda}P_A(a_x|x,\lambda)\pi_{\lambda}{\sigma}_{\lambda}^B
\end{align}
where, $\pi_{\lambda}=\tr[(J_{\lambda}\otimes\Id_B)\rho^{AB}]$ and $\sigma^B_{\lambda}=\frac{1}{\pi_{\lambda}}\tr_A[(J_{\lambda}\otimes\Id_B)\rho^{AB}]$. Hence, $\rho^{\prime AB}$ is not steerable from $A$ to $B$ irrespective of the choice of the measurement assemblage $\cM_A=\{M_x\}^n_{x=1}$.
\end{proof}
From this result, the next corollary is immediate:
\begin{corollary}\label{obs-incom-break}
Suppose Bob is sending $A$ part of a bipartite state $\rho^{AB}$ to Alice (i.e., from $B$ to $A$) through an incompatibility breaking channel $\Lambda$ then $\rho^{\prime AB}=(\Lambda\otimes\cI)(\rho^{AB})$ is not steerable from $A$ to $B$ where $\cI$ is the identity channel. 
\end{corollary}
Therefore, if Bob has only IBCs to communicate with Alice, the resulting bipartite state after the communication will no longer be steerable from $A$ to $B$. Corollary \ref{obs-incom-break} is also independently proved in the Theorem 1 of \cite{Franco-steer}. A similar corollary, but not exactly same as the corollary in \ref{obs-incom-break}, has been derived using channel state duality in \cite{Kiukas1}. But if Bob has a quantum switch, he can get rid of this situation. It will be clarified through the following example.
\begin{example}
\rm{Suppose Bob is sending Alice the $A$ part of a quantum state $\rho^{AB}=\ket{\psi}_{AB}\bra{\psi}$, where $\ket{\psi}_{AB}=\frac{1}{\sqrt{2}}[\ket{00}+\ket{11}]$. $\ket{\psi}_{AB}$ is steerable under suitable measurements. Clearly, $\frac{1}{\sqrt{2}}\mid\tr[\rho^{AB}(\sigma_x\otimes\sigma_x)]+\tr[\rho^{AB}(\sigma_z\otimes\sigma_z)]\mid=\sqrt{2}> 1$. Now suppose Bob has only IBCs, for example, the quantum channel $\Phi$, given in equation \eqref{good_chan} to communicate with Alice. Since, $\Phi$ is EBC (Therefore it is also IBC). Therefore, $(\Phi\otimes\cI)(\rho^{AB})$ is separable and so it is not steerable from Alice to Bob i.e., from $A$ to $B$. Now if Bob has a quantum switch he can improve the communication. After using quantum switch, the effective channel $\cC_{eff}(\rho)=Tr_{\cH_c}[\cU S_{\Phi,\Phi,\omega}(\rho)\cU^{\dagger}]$ is given in equation \eqref{good_chan_switched}. Therefore, the resulting shared state after the communication is
\begin{align}
\rho^{\prime AB}=&(\cC_{eff}\otimes\cI)(\rho^{AB})\nonumber\\
=&\frac{1}{2}[\cC_{eff}(\ket{0}\bra{0})\otimes\ket{0}\bra{0}+\cC_{eff}(\ket{0}\bra{1})\otimes\ket{0}\bra{1}\nonumber\\
&+\cC_{eff}(\ket{1}\bra{0})\otimes\ket{1}\bra{0}+\cC_{eff}(\ket{1}\bra{1})\otimes\ket{1}\bra{1}]\nonumber\\
=&\frac{(1+\lambda_3)^2}{4}\ket{\psi}_{AB}\bra{\psi}+\Big(1-\frac{(1+\lambda_3)^2}{4}\Big)\rho_{sep}^{AB}
\end{align}
where $\rho_{sep}^{AB}=\frac{2}{4-(1+\lambda_3)^2}\Big[\frac{(1-\lambda_3)^2}{4}\ket{00}\bra{00}\\
+\frac{1-\lambda^2_3}{2}\ket{01}\bra{01}+\frac{1-\lambda_3^2}{2}\ket{10}\bra{10}+\frac{(1-\lambda_3)^2}{4}\ket{11}\bra{11}\Big]$ is a separable state.

Now,  $\frac{1}{\sqrt{2}}\mid\tr[\rho^{AB}_{sep}(\sigma_x\otimes\sigma_x)]+\tr[\rho^{AB}_{sep}(\sigma_z\otimes\sigma_z)]\mid=\frac{3\lambda_3^2-2\lambda_3-1}{4-(1+\lambda_3)^2}$.

Therefore,

\begin{align}
F(\rho^{\prime AB})&=\frac{1}{\sqrt{2}}\mid\tr[\rho^{\prime AB}(\sigma_x\otimes\sigma_x)]+\tr[\rho^{\prime AB}(\sigma_z\otimes\sigma_z)]\mid\nonumber\\
&=\frac{(1+\lambda_3)^2}{4}\sqrt{2}+ \frac{3\lambda_3^2-2\lambda_3-1}{4}.
\end{align}
\begin{figure}[hbt!]
\centering
\includegraphics[scale=0.48]{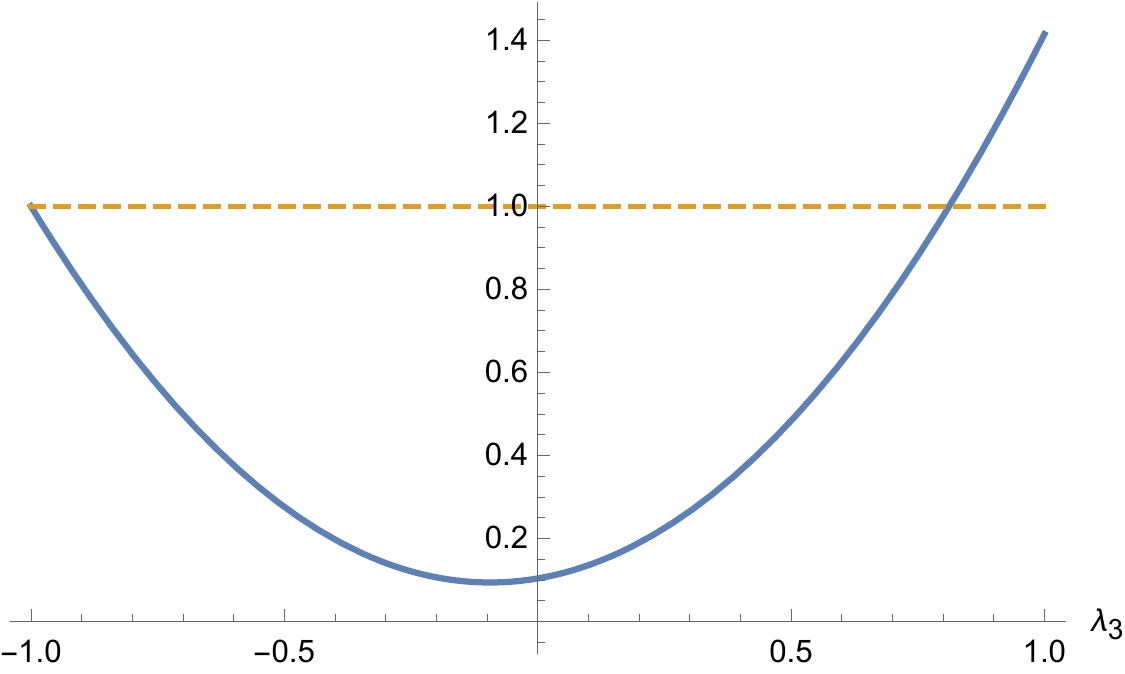}
\caption{$F(\rho^{\prime AB})$ vs. $\lambda_3$ graph. $F(\rho^{\prime AB})>1$ for $\lambda_3>0.8123$.  }
\label{fig:SteeringInequality}
\end{figure}

From the Fig. \ref{fig:SteeringInequality}, we  get that $F(\rho^{\prime AB})>1$ for $1\geq\lambda_3>0.8123$. Therefore, in this range of $\lambda_3$, $\cC_{eff}$ is not $2$-IBC. As $0.8123>2^{\frac{3}{4}}-1$, from Example \ref{ex_prob_adv_qrac}, we guess that possibly QRAC is a better detector of $n$-IBCs than steering inequalities at least in some cases.}

\end{example}

\subsubsection{Prevention of the loss of coherence}
Quantum coherence is an important resource that can be utilized in several information-theoretic and thermodynamic tasks \cite{res-coh}. Specifically, it can be used in work extraction \cite{Korzekwa},  quantum algorithms \cite{Hillery}, quantum metrology \cite{Giorda}, quantum channel discrimination \cite{Piani},  witnessing quantum correlations \cite{Bu} etc.\\
We know from the resource theory of coherence \cite{res-coh} that there are operations, known as incoherent operations which do not increase coherence but may decrease it and there are channels, known as coherence breaking channels that completely destroy the coherence of any quantum state.\\
Now suppose Alice is preparing quantum states for Bob who uses the coherence of these states w.r.t. some basis as a resource to get the advantage in different information-theoretic and thermodynamic tasks. But if Alice has only coherence-breaking channels to communicate with Bob, she will be unable to transfer the coherence of the state. Whereas if she has incoherent channels, the coherence of the state which Alice sends for Bob, will decrease when the state will reach Bob.\\
In these types of cases, if she has a quantum switch, she might be able to perform better communication, as stated in the following theorem.\\

\begin{theorem}
A coherence breaking qubit channel may be converted to a non-coherence breaking qubit channel with the help of a quantum switch along with a measurement-based controlled incoherent unitary operation.
\end{theorem}
\begin{proof}
Suppose we are interested in coherence w.r.t. the basis $\{\ket{0}, \ket{1}\}$ of a qubit.
 Then, the $T$ matrix of any coherence breaking channel $\mathcal{E}$ has the following form \citep{Pati18}:
 \begin{align}
\mathcal{T}_{\mathcal{E}} = 
\begin{pmatrix}
1 & 0 & 0 & 0\\
0 & 0& 0&0\\
 0 & 0 & 0&0\\
t&0&0&  \lambda\\
\end{pmatrix}
\end{align}
the Kraus operators (written in the eigen basis of $\sigma_z$) for this channel are given as follows
\begin{equation}
\begin{aligned}
	K_1 &=\sqrt{\dfrac{1}{2}(1-\lambda-t)}
\begin{pmatrix}
0 & 0\\
1 & 0
\end{pmatrix}\\
	K_2 &=\sqrt{\dfrac{1}{2}(1+\lambda-t)}
\begin{pmatrix}
0 & 0\\
0 & 1
\end{pmatrix}\\
	K_3 &=\sqrt{\dfrac{1}{2}(1-\lambda+t)}
\begin{pmatrix}
0 & 1\\
0 & 0
\end{pmatrix}\\
K_4 &=\sqrt{\dfrac{1}{2}(1+\lambda+t)}
\begin{pmatrix}
1 & 0\\
0 & 0
\end{pmatrix}\\
\end{aligned}
\end{equation}
Let us start  with the most general qubit state (written in the eigen basis of $\sigma_z$) given by
\begin{equation}
	\rho=\dfrac{1}{2}\begin{pmatrix}
1+c & a-ib\\
a+ib & 1-c
\end{pmatrix}
\end{equation}
Given the Krauss operators of the channel, we can calculate the channel under the switch using   equation \eqref{switched_channel}. Let the final state (after the use of the switch) be represented by $\mathsf{S}_{\mathcal{E},\omega}(\rho)$, where $\mathcal{E}$ is the coherence breaking channel. If we trace over the control part, the resulting channel naturally has the form of the initial channel and we get no advantage as we can see from the final state after tracing out the control qubit.

\begin{align}
\Gamma^{\prime}_{eff}(\rho)&=\tr_{\cH_c} [\mathsf{S}_{\mathcal{E},\mathcal{E},\omega}(\rho) ] \nonumber\\
&=\dfrac{1}{2}\begin{pmatrix}
1+c\lambda^2+t(1+\lambda) & 0\\
0 & 1-c\lambda^2-t(1+\lambda) 
\end{pmatrix}
\end{align}

The T matrix of this effective channel is therefore given in the following coherence breaking form
\begin{equation}
\begin{aligned}
\mathcal{T}_{\Gamma^{\prime}_{eff}}&=\begin{pmatrix}
1 & 0 & 0 & 0\\
0&0 & 0 & 0\\
0& 0 & 0 & 0\\
t(1+\lambda)& 0 & 0 & \lambda^2
\end{pmatrix}.
\end{aligned}
\end{equation}
But, instead if we use a general measurement-based incoherent controlled unitary operation before tracing out the control qubit we can in-fact get an effective non-coherence breaking channel. Let the measurement-based controlled unitary be given by  $\cU=I\otimes\ket{+}\bra{+}+U\otimes\ket{-}\bra{-}$ where
\begin{equation}
U=\begin{pmatrix}
e^{i\phi_1}\cos\theta & e^{i\phi_2}\sin\theta\\
-e^{-i\phi_2}\sin\theta& e^{-i\phi_1}\cos\theta
\end{pmatrix}.
\end{equation}

Here the matrix representation of $U$ is written in the eigen basis of $\sigma_z$. The final qubit state after the use of the switch followed by the controlled unitary and the partial trace over the control qubit is given by the following:
\begin{align}
&\Gamma_{eff}(\rho)\nonumber \\
&=\tr_{\cH_c} [\cU\mathsf{S}_{\mathcal{E},\mathcal{E},\omega}(\rho)\cU^\dagger ] \nonumber \\
&=\dfrac{1}{2}\begin{pmatrix}
1+t_1+c\eta_3+a\eta_1-ib\eta_2&t_2^*+c\gamma_3^*+a\gamma_1^*-ib\gamma_2^*\\
t_2+c\gamma_3+a\gamma_1+ib\gamma_2& 1-t_1-c\eta_3-a\eta_1+ib\eta_2
\end{pmatrix}
\\
&\text{where,}\\
 &t_1=(1+\lambda)(6+2\cos\theta)t/8,\\
&\eta_1=[-\cos(\phi_1-\phi_2)\sin2\theta (1-2\lambda+\lambda^2-t^2)]/8,\\
&\eta_2=[\sin(\phi_1-\phi_2)\sin2\theta (1-2\lambda+\lambda^2-t^2)]/8, \\
&\eta_3=(1+2\lambda+t^2)(1-\cos^2 2\theta)/8+\lambda^2\cos2\theta\\
&\gamma_1= (1-2\lambda+\lambda^2-t^2)(1-e^{-2i\phi_1}\cos\theta^2+e^{-2i\phi_2}\sin\theta^2)/8\\
&\gamma_2= (1-2\lambda+\lambda^2-t^2)(1-e^{-2i\phi_1}\cos\theta^2-e^{-2i\phi_2}\sin\theta^2)/8\\
&\gamma_3= (1-2\lambda+\lambda^2-t^2)e^{-i(\phi_1+\phi_2)}\sin 2\theta\\
&t_2=-2e^{-i(\phi_1+\phi_2)}(\lambda+1)t\sin 2\theta.
\end{align}

The $\mathcal{T}$ matrix of the effective channel is given as follows
\begin{equation}
\begin{aligned}
	\mathcal{T}_{\Gamma_{eff}}&=\begin{pmatrix}
1 & 0 & 0 & 0\\
t_2^R&\gamma_1^R & -\gamma_2^I & \gamma_3^R\\
t_2^I&\gamma_1^I & \gamma_2^R & \gamma_3^I\\
t_1& \eta_1 & -i\eta_2 & \eta_3
\end{pmatrix}\\
&\text{where} \hspace{1cm} \text{$x^R=Re(x)$  and $x^I=Im(x)$}
\end{aligned}
\end{equation}
 The  case of $\theta=0$ corresponds to the incoherent unitary operator, for which the T-matrix takes the following form:
\begin{equation}
\begin{aligned}
 \mathcal{T}_{\Gamma_{eff}(\theta=0)}&=\begin{pmatrix}
1 & 0 & 0 & 0\\
0&\gamma_1^R &  -\gamma_2^I & 0\\
0& \gamma_1^I & \gamma_2^R & 0\\
t_1& 0 & 0 & \eta_3
\end{pmatrix}
\end{aligned}
\end{equation}
Here, for $\phi_1=\frac{\pi}{2}$, $\eta_3=\lambda^2$, we have $t_1=t(1 + \lambda )/2$ and $\gamma_1^R=\gamma_2^R=(1 - 2  \lambda+ \lambda^2 - t^2)/4$. Therefore, using an incoherent controlled unitary, a quantum switch can be used to convert a coherence breaking channel into a non-coherence breaking channel. However in this particular case, since the components $\mathcal{T}_{13}$, $\mathcal{T}_{23}$, $\mathcal{T}_{31}$, and $\mathcal{T}_{32}$ are zero, the effective channel after the switch operation can not generate coherence but can reduce the loss.
\end{proof}

 \subsection{ Communication using noisy quantum switch}\label{subsec:noisy switch}
In the previous sections, we have assumed that the control qubit is not interacting with the environment. But in practice the control qubit will interact with its environment which will introduce noise in the state of the quantum switch. Therefore, through measurement-based controlled operation one may not get the desired effective channel. Now, we will study this through an example. 
\begin{example}[Switch with depolarising noise] 
\rm{Consider, the quantum channel $\Lambda_{perfect}(\rho)=\frac{1}{2}\sigma_x\rho\sigma_x+\frac{1}{2}\sigma_y\rho\sigma_y$. We know that perfect communication can be achieved through this channel using quantum switch \citep{chiribella_perfect}. 
Therefore,

\begin{align}
\mathsf{S}_{\Lambda_{perfect},\Lambda_{perfect},\omega}(\rho)=\frac{1}{2}\rho\otimes\omega+\frac{1}{2}\sigma_z\rho\sigma_z\otimes\sigma_z\omega_z\sigma_z
\end{align}
where, $\omega=\ket{+}\bra{+}$.\\
After this, suppose the depolarising, noisy channel $\Gamma^{t}_2$ has acted on the control qubit state $\omega$ due to the interaction with the environment. Then the joint state of the system and the switch is given by
\begin{align}
\mathsf{S}^{\prime}_{\Lambda_{perfect},\Lambda_{perfect},\omega}(\rho)&=\frac{1}{2}\rho\otimes\Gamma^{t}_2(\omega)+\frac{1}{2}\sigma_z\rho\sigma_z\otimes\Gamma^{t}_2(\sigma_z\omega\sigma_z)\nonumber\\
&=\frac{1}{2}(\frac{(1+t)}{2}\rho+\frac{(1-t)}{2}\sigma_z\rho\sigma_z)\otimes\ket{+}\bra{+}\nonumber\\
&+\frac{1}{2}(\frac{(1+t)}{2}\sigma_z\rho\sigma_z+\frac{(1-t)}{2}\rho)\otimes\ket{-}\bra{-}.
\end{align}
Now suppose Alice is implementing a measurement based controlled operation  $\cU=\Id\otimes \ket{+}\bra{+}+\sigma_z\otimes\ket{-}\bra{-}$
Then the final state is given by
\begin{align}
Tr_{\cH_c}[\cU\mathsf{S}^{\prime}_{\Lambda_{perfect},\Lambda_{perfect},\omega}(\rho)\cU^{\dagger}]=\frac{(1+t)}{2}\rho+\frac{(1-t)}{2}\sigma_z\rho\sigma_z.
\end{align}

Therefore, the above channel is no longer an identity channel, and hence, perfect communication is not achieved. In a similar way, it can be easily shown that the noisy quantum switch may hamper improvement in communication through other quantum channels.}

\end{example}

\section{conclusion}\label{sec:conc}
In this work, we have presented several results regarding the improvement in quantum communication using the quantum switch. {We have argued that it is important to study the effect of the quantum switch by studying one output branch at a time. We have shown that if a useless channel remains useless even after using the quantum switch, concatenating it with  another channel and subsequently using the quantum switch may provide improvements in communication. In particular, we have studied the conversion of \emph{completely useless} channels into useful channels through concatenation and use of quantum switch. This result might be useful in quantum communication technology in the future.} We have shown that  improvements in communication due to the action of the quantum switch help us (i) to get the advantage in Quantum Random Access Codes as well as to demonstrate the quantum steering when only useless channels are available for communication, (ii) to prevent the loss of coherence, etc. We have shown that noise introduced in the switch may hamper the communication improvement.\\
Our work opens up several research avenues. It is an open problem to find out the necessary-sufficient condition for a generic quantum channel that can provide improvement under the action of the quantum switch. Though we have shown that if a channel is useless even after using the quantum switch, concatenating it with another channel may provide improvement in communication under the action of the quantum switch, the necessary and sufficient condition for this improvement, in this case, is not known. It may also be interesting to compare the effectiveness of the noisy quantum switch in achieving improvements in different quantum information processing (or quantum communication) tasks. A study of the communication problems, considered here in the context of quantum switch, in some other related contexts (like SDPP, etc.) is worth pursuing. Moreover, a comparison of the results -- obtained here using the quantum switch and the ones that might be obtained using SDPP -- is also an important aspect, to be considered in near future. Another interesting future direction is the application of quantum switch in the context of quantum dynamical maps to see whether, in particular, action of such a switch can improve the speed of communication.

\section{Acknowledgements}

We would like to thank Huan-Yu Ku for the valuable comments on this work.

\end{document}